\documentclass{llncs}
\usepackage{fullpage}
\usepackage[T1]{fontenc}
\usepackage{amsmath,amsfonts,amssymb,amstext}
\usepackage[usenames,dvipsnames]{xcolor}
\definecolor{DarkGreen}{rgb}{0.0,0.5,0.0}
\usepackage{url}
\usepackage{booktabs}
\usepackage{enumerate}
\usepackage{tikz}
\usepackage[colorlinks=true,linkcolor=MidnightBlue,citecolor=DarkGreen,urlcolor=Maroon]{hyperref}
\makeatletter
\g@addto@macro\appendix{}
\makeatother

\newcommand{\naturals}{\mathbb{N}}
\newcommand{\integers}{\mathbb{Z}}

\renewcommand{\d}[1]{\ensuremath{\operatorname{d}\!{#1}}}

\newcommand{\ignore}[1]{}

\newcommand{\E}{\mathbb{E}}

\title{Oblivious Self-Distance Symmetric Rendezvous on the Integer Line\thanks{This paper is the full version of a work accepted to the 37th International Symposium on Algorithms and Computation (ISAAC 2026), to be held December~6--9,~2026, in Hangzhou, China~\cite{GeorgiouGravelMandicKapusinISAAC26}.}}
\author{
Konstantinos Georgiou\orcidID{0009-0007-9677-341X}\thanks{Research supported in part by an NSERC Discovery Grant and by the Toronto Metropolitan University Faculty of Science Dean's Research Fund. \texttt{konstantinos@torontomu.ca}}
\inst{1}
\and
Claude Gravel\thanks{\texttt{gravel@torontomu.ca}}
\inst{2}
\and
Joey Kapusin\thanks{\texttt{jkapusin@torontomu.ca}}
\inst{1}
\and
Lazar Mandic\thanks{\texttt{lazar.mandic@torontomu.ca}}
\inst{1}
}

\institute{
Department of Mathematics, Toronto Metropolitan University, Toronto, Ontario, Canada
\and
Department of Computer Science, Toronto Metropolitan University, Toronto, Ontario, Canada
}

\date{}
\begin{document}
\hypersetup{pageanchor=false}
\maketitle
\thispagestyle{empty}

\begin{abstract}

Symmetric rendezvous on the line is a classical search problem in which two agents, initially placed at distance $2d$, must follow the same randomized strategy in order to meet as quickly as possible. In the standard model, agents may condition their actions on the entire execution history, and both the known- and unknown-distance variants admit expected rendezvous time $\Theta(d)$. We study the role of memory by introducing oblivious self-distance strategies, in which an agent's decision depends only on her current position relative to her own starting location.

For an initial separation of $2d$, let $R_d$ denote the optimal oblivious expected rendezvous time in the known-distance setting. We develop two complementary finite-state frameworks based on absorbing Markov chains. Truncated chains give computable upper bounds through finite-support strategies, while weak-peek chains give lower bounds through a revealed-information relaxation. Together, they provide a mechanism for certifying optimality. Using that mechanism, we determine $R_1$ exactly and prove that it is attained by a finite-support strategy. For $d=2,\ldots,6$, extensive numerical optimization repeatedly gives the same truncation structure and objective values, yielding rigorous upper bounds below $7.83d^2$. We do not prove that the computed weak-peek minimizers are global, but based on the stability of the computations we conjecture that they are, in which case the corresponding truncated strategies are optimal.

We also prove that $R_d=\Theta(d^2)$. In the unknown-distance setting, we construct a universal strategy, independent of $d$, with expected rendezvous time $O(d^{2+\eta})$ for every fixed $\eta>0$. Thus, under the memory restriction, the known-distance rendezvous time becomes quadratic, while near-quadratic performance remains possible even without knowing $d$. The asymptotic analysis uses birth-death Markov chains and their electrical-network interpretation.

\end{abstract}

\subsubsection*{Acknowledgements}
The authors used ChatGPT as an auxiliary tool for exploratory discussion, code generation and debugging for preliminary experiments, and consistency checks of intermediate arguments. All mathematical claims, proofs, computational results, and final verification are the sole responsibility of the authors.

\newpage
\thispagestyle{empty}
\tableofcontents

\newpage
\pagenumbering{arabic}
\hypersetup{pageanchor=true}

\section{Introduction}

Symmetric rendezvous on the line, introduced by Alpern~\cite{alpern1995rendezvous}, is one of the most extensively studied rendezvous problems in search theory. 
Two agents are initially placed on the infinite line at distance $2d$, but neither agent knows on which side the other resides. In every step, the agents execute synchronously the same randomized strategy and independently realize her random choices, with the objective of minimizing the expected rendezvous time. This symmetry requirement fundamentally distinguishes the problem from its asymmetric counterpart, where different roles may be assigned to the agents and optimal rendezvous is known.

In the known-distance setting, the value of $d$ is known to both agents. Han et al.~\cite{han2008improved} prove that optimal strategies may be restricted to movements in steps of length $d$. This justifies the standard normalization $d=1$, where the agents start at distance $2$.
Despite decades of research, the exact optimal expected rendezvous time remains unknown. The best lower and upper bounds are due to Han et al.~\cite{han2008improved}, who proved
$
4.152d \leq OPT \leq 4.2574d,
$
and conjectured that the true relative value is close to $4.25$. In the unknown-distance setting, the agents know neither the direction nor the value of $d$.
The best bounds currently known imply
$
4.152d \leq OPT \leq 13.93d,
$
where the upper bound is due to Klimm et al.~\cite{klimm2022competitive}.

Thus, both the known-distance and unknown-distance variants admit optimal expected rendezvous time linear in the initial separation. However, the exact optimal constants remain unknown. Moreover, the best known strategies rely on increasingly sophisticated randomized movement patterns. In the known-distance setting, the strongest upper bounds are obtained through block strategies of increasing support size, while in the unknown-distance setting the best algorithms exploit growing excursions together with nontrivial dependence on previous decisions. This raises a natural question: \emph{what role does memory play in rendezvous on the line?}

\subsection{Model Motivation} 
One natural way to investigate the role of memory is to progressively restrict the information available to the agents. In the classical rendezvous problem, an agent may condition her decisions on the entire execution history, including all previous movements and random choices. In contrast, we consider a substantially more restrictive setting. At every step, an agent knows only her current distance from her own origin together with the direction of that origin. She does not remember the sequence of previous movements that produced her current position. 
Conceptually, this can also be understood as the agent observing her local state (depending only on her relative location to the initial placement) and applying the same fixed stochastic decision rule.
We refer to such strategies as \emph{oblivious self-distance strategies}.

Our model may be viewed as complementary to the revealed-information framework of Alpern~\cite{Alpern2007Revealed}. There, additional information is provided to the agents during the execution of the strategy, leading to lower bounds for the classical rendezvous problem. Here, we move in the opposite direction, asking how much rendezvous performance can be achieved when information is removed rather than added. In this sense, the oblivious model provides a natural framework for quantifying the value of memory in symmetric rendezvous on the line.

The restriction is also motivated by the structure of the known algorithms. As discussed above, the best known strategies rely on increasingly sophisticated movement patterns and increasingly large supports. Consequently, both the analysis and the computation of optimal strategies become rapidly intractable. By restricting attention to oblivious self-distance strategies, we obtain a model that remains rich enough to exhibit nontrivial rendezvous behaviour, while admitting a finite-state representation that can be analyzed systematically.

\subsection{Discussion on Results \& Significance}
For an initial agent separation of $2d$, with $d\in\naturals$, let $R_d$ denote the optimal expected rendezvous time in the known-distance setting. Since the known-distance problem provides additional information to the agents, $R_d$ is a lower bound on the expected rendezvous time at distance $2d$ of every strategy admissible in the unknown-distance setting. Our results provide a detailed picture of both settings, revealing several fundamental differences between classical and oblivious rendezvous on the line.

First, unlike the classical known-distance problem, the oblivious model does not admit a scaling reduction. Indeed, agents' movements are of step-size 1. However, in the unrestricted setting, agents can employ their memory in order to apply step-size $d$ movements, effectively reducing the distance-$2d$ problem to the distance-$2$ problem. Equivalently, the classical known-distance problem is invariant under scaling, and therefore the optimal value is determined by a single constant, independently of $d$. In the oblivious setting, however, movements remain of step-size 1, and decisions depend explicitly on an agent's current distance from her origin. Consequently, changing $d$ changes the underlying optimization problem, and results for one value of $d$ do not extend to another. In particular, determining the exact value of $R_1$ does not determine the value of $R_d$ for larger values of $d$.

A second contribution concerns the structure of optimal strategies. In the classical rendezvous problem, the best known upper bounds are obtained through strategies of increasingly large support, and there is substantial evidence that optimal behaviour is not captured by finite-support strategies. In contrast, our results suggest a fundamentally different picture in the oblivious setting. We determine the exact value of $R_1$ and prove that it is attained by an optimal finite-support strategy. For $d=2,\ldots,6$, extensive numerical optimization repeatedly identifies the same finite-support structure and essentially the same objective values. The resulting truncated strategies give rigorous upper bounds, while global optimality of the corresponding weak-peek minimizers remains unproved. Based on the numerical stability, we conjecture that the reported strategies are optimal.
Thus, while restricting memory increases the expected rendezvous time from linear to quadratic, it appears to simplify the structure of optimal strategies. 
Beyond the exact computations, we construct bounded-support strategies with expected rendezvous time below $9d^2$ for every $1\leq d\leq500$, and also for $d=1000$.

A central contribution of the paper is a finite-state framework for studying oblivious rendezvous. The framework consists of two complementary components. The first produces computable upper bounds by restricting attention to finite-support strategies, while the second produces computable lower bounds through a suitable relaxation of the rendezvous problem. Together, the two frameworks provide a certification mechanism capable of proving optimality. More broadly, the framework converts an infinite-state rendezvous problem into a finite-state optimization problem, yielding new lower-bound techniques, rigorous computational methods, and a systematic approach for studying memory-constrained rendezvous.

Finally, we establish an asymptotic characterization of oblivious rendezvous on the line. We prove that $R_d=\Theta(d^2)$, and propose a universal rendezvous strategy with expected rendezvous time $O(d^{2+\eta})$, for every fixed $\eta>0$. Thus, restricting agents to oblivious self-distance strategies changes the complexity of rendezvous from linear to quadratic. Perhaps more surprisingly, the unknown-distance problem still admits rendezvous in nearly quadratic time through an infinite-support oblivious strategy, showing that the known-distance and unknown-distance settings remain nearly asymptotically equivalent even under severe memory restrictions. These results establish a qualitative separation between unrestricted-memory and oblivious rendezvous on the line, identifying memory as a fundamental resource governing both the complexity of rendezvous and the structure of optimal strategies.

\subsection{Related Work}

Rendezvous search is a classical topic in search theory and search games. Early foundational work includes the rendezvous problem on discrete locations~\cite{AW90-Discrete} and the general rendezvous-search framework introduced by Alpern~\cite{alpern1995rendezvous}. The area has since developed into a broad research direction, with surveys and monographs given in~\cite{A02-Perspective},~\cite{alpern2013search} and~\cite{alpern2006theory}. Since the literature is extensive, we mention only a few representative results and directions most relevant to the present work.

A variety of rendezvous models have been studied under different assumptions regarding the search domain, the available information, and the capabilities of the agents. Representative examples include rendezvous on 
complete graphs~\cite{cembrano2026faster}, rendezvous in rings~\cite{KDM10-RingBook}, rendezvous with advice~\cite{georgiou2019symmetric}, multi-robot rendezvous on the line~\cite{lim1997rendezvous,ozsoyeller2021multi}, rendezvous with marks left at starting locations~\cite{BG01-Marks}, rendezvous with marks dropped at chosen times~\cite{leone2018rendezvous}, and rendezvous with revealed information~\cite{Alpern2007Revealed}. Computational and symbolic approaches to rendezvous search have also been investigated~\cite{AL22-Symbolic}.

Among continuous rendezvous models, the line is one of the most extensively studied domains. 
Classical line rendezvous models distinguish between distinguishable players, who may employ different (asymmetric) strategies~\cite{alpern1995rendezvous-dist}, and indistinguishable players, who must employ the same (symmetric) randomized strategy~\cite{anderson1995rendezvous}.
Asymmetric rendezvous on the line was later related to double linear search~\cite{AG99-AsymmetricDouble}. Other line variants include rendezvous on the labeled line~\cite{CT04-Labeled}, rendezvous with revealed information~\cite{Alpern2007Revealed}, and rendezvous with marks~\cite{BG01-Marks,leone2018rendezvous}. Variants where the initial distance is known~\cite{anderson1995rendezvous,han2008improved,B99-TwoRend} and unknown~\cite{beveridge2013symmetric,klimm2022competitive} have both received considerable attention. Representative improvements for the symmetric known-distance problem were obtained by Anderson and Essegaier~\cite{anderson1995rendezvous}, Baston~\cite{B99-TwoRend}, and Han et al.~\cite{han2008improved}, while the unknown-distance setting was studied by Beveridge et al.~\cite{beveridge2013symmetric} and later improved by Klimm et al.~\cite{klimm2022competitive}. Despite substantial progress, the exact symmetric rendezvous value on the line remains unknown.

The present paper remains within the classical symmetric rendezvous on the line framework, but focuses on oblivious strategies and considers both the known-distance and unknown-distance settings.
Similar notions of oblivious agents appear in the mobile-robot literature, e.g. in~\cite{flocchini2026universal}, where agents act solely on the basis of the current observation and retain no memory of previous executions. Our work investigates how this memory restriction affects the complexity and performance of symmetric rendezvous on the line.

\subsection{Paper Organization and Summary of Results}

We begin in Section~\ref{sec:DefResults} by introducing the oblivious self-distance rendezvous model, defining the known-distance and unknown-distance settings, and presenting our main results. Section~\ref{sec:MarkovianReformulation} then reformulates the rendezvous process as an absorbing Markov chain, providing the probabilistic framework used throughout the paper.

Section~\ref{sec:SmallDFrameworks} develops our finite-state certification framework for small values of $d$. In Section~\ref{sec:TruncatedChains}, we introduce truncated chains and show how they yield computable upper bounds through finite-support strategies. Section~\ref{sec:PeekChains} introduces weak-peek chains, which provide computable lower bounds through a finite aided relaxation of the original problem. Section~\ref{sec:AllBoundsForSmallD} combines the two frameworks through a certificate mechanism that allows matching upper and lower bounds to be established whenever an optimal weak-peek solution exhibits certain structural properties. We then apply this framework in Section~\ref{sec:UpperBoundsSmallD1} to determine the exact value of $R_1$ and prove optimality of the corresponding truncated strategy.
Section~\ref{sec:NumericalComputations} presents extensive computations for $d=2,\ldots,6$. The resulting truncated strategies give rigorous finite-support upper bounds, while the numerical computations lead us to conjecture that the corresponding strategies are optimal.

Section~\ref{sec:AsymptoticBounds} turns to asymptotic values of $d$. Section~\ref{sec:ElectricalRepresentationInducedChains} develops the electrical-network machinery used in the remainder of the paper. Section~\ref{sec:QuadraticLowerBounds} establishes a universal quadratic lower bound for all oblivious self-distance strategies, including the known-distance setting. Section~\ref{sec:QuadraticUpperBound} proves an $O(d^{2+\eta})$ upper bound in the unknown-distance setting by constructing a single infinite-support strategy that is independent of $d$. Section~\ref{sec:AbsQuadraticUpperBound} establishes a pure quadratic upper bound in the known-distance setting through a family of truncated strategies. Finally, Section~\ref{sec:NumEstimatesQuadraticConstants} presents numerical evidence on the asymptotic truncation ratio and the associated quadratic rendezvous constant.
Finally, our paper concludes with a discussion on future directions in Section~\ref{sec:Conclusion}.
For a quick reference, we end this section with a summary of our main results presented in Table~\ref{tab:main-results}.

\begin{table}[!h]
\centering
\small
\begin{tabular}{p{0.78\linewidth}c}
\hline
Result & Reference \\
\hline
Truncated-chain upper-bound framework
&
Lemma~\ref{lem:UpperBoundFrame}
\\
Weak-peek chain lower-bound framework
&
Lemma~\ref{lem:WeakPeekLowerBound}
\\
Finite-state optimality certification mechanism
&
Lemma~\ref{lem:WeakPeekCertificate}
\\
Known-distance optimal result:
$R_1\approx6.467098
				$
&
Theorem~\ref{thm:optR1}
\\
Finite-support upper bounds for $d=2,\ldots,6$ (conjectured optimality)
&
Section~\ref{sec:NumericalComputations},
Table~\ref{tab:weakpeek-summary}
\\
Known-distance asymptotics:
$R_d=\Theta(d^2)$
&
Theorems~\ref{thm:QuadraticLowerBound}
and~\ref{thm:PureQuadraticTruncatedUniform}
\\
Unknown-distance asymptotics:
There exists a universal strategy, independent of $d$, with expected rendezvous time
$O(d^{2+\eta})$ for every fixed $\eta>0$
&
Theorems~\ref{thm:QuadraticLowerBound}
and~\ref{thm:NearQuadraticUpperBound}
\\
Numerical estimates of asymptotic truncation constants
&
Section~\ref{sec:NumEstimatesQuadraticConstants},
Table~\ref{tab:uniform-truncation-summary}
\\
\hline
\end{tabular}
\caption{Summary of the main results.}
\label{tab:main-results}
\end{table}

\section{Definitions and Preliminary Observations}
\label{sec:DefResults}

\subsection{Problem Definition}
\label{sec:ProblemDefinition}

We study symmetric rendezvous problems on the infinite path $\mathbb Z$ under a restricted class of randomized strategies. We focus on \emph{oblivious self-distance strategies}, where each agent selects her move based only on her current distance from her own starting position, and not on the history of the process.

We consider two agents moving synchronously on the integer line $\mathbb Z$. Their initial positions are at distance $2d$, for some integer $d\geq 1$. For the analysis, we place their initial positions at $-d$ and $d$.
Time proceeds in discrete synchronous rounds $t=0,1,2,\dots$. Let $X_t, Y_t \in \integers$ denote the displacements of the two agents from their respective starting positions. Thus their actual positions at time $t$ are
$-d+X_t$
and 
$d+Y_t$, 
with
$X_0=0$, and $Y_0=0$.
These coordinates are used only for the analysis. The agents do not share a common orientation of the line, that is, they cannot distinguish the two global directions ``left'' and ``right''. The only directional distinction available to an agent away from her starting position is whether a move increases or decreases her distance from that starting position.

At each time step, each agent moves by one unit, and hence
$X_t-X_{t-1}, Y_t-Y_{t-1}\in\{-1,+1\}$, 
for every $t\geq 1$.
Rendezvous occurs when the two agents occupy the same vertex at the same time, i.e.
$
-d+X_t=d+Y_t,
$
or
$
X_t=Y_t+2d.
$
The corresponding \emph{rendezvous time} is defined as the random variable
\begin{equation}
\label{equa:RendOriginal}
\mathcal T_d:=\inf\left\{t\geq 0:X_t=Y_t+2d\right\}.
\end{equation}

An \emph{admissible strategy} is a single randomized algorithm that is used by both agents.
At each time $t$, the next move of an agent is determined by a probability distribution that depends only on her current distance from her own starting position. Thus the algorithm assigns, for each $i\geq 0$, a probability of moving away from the starting point when the current distance is $i$.
The strategy is \emph{symmetric} in the sense that both agents use the same distribution. Conditional on their current distances from their respective starting positions, the two agents sample their moves independently according to this distribution.
Accordingly, an admissible strategy is identified with a sequence
$
(a_i)_{i\geq 0},
$
where $0\leq a_i\leq 1$ denotes the probability that an agent at distance $i\geq 1$ from her starting point moves away from that starting point. Consequently, $1-a_i$ is the probability that the agent moves toward her starting point.  
At distance $0$, both possible moves increase the distance from the
starting point from $0$ to $1$. Since the two directions are indistinguishable
under the symmetry of the model, the move from the origin must be unbiased.
For notational convenience, and to allow a uniform description of the
transition probabilities in the Markovian reformulation of Section~\ref{sec:MarkovianReformulation} below, we set
$
a_0=\frac12.
$

The objective is to minimize the \emph{expected rendezvous time}. For a fixed admissible strategy $a=(a_i)_{i\geq 0}$, let $\mathcal T_d=\mathcal T_d(a)$ denote the rendezvous time defined in~\eqref{equa:RendOriginal}. 
There are two natural versions of the problem with respect to the set of admissible strategies. In the \emph{known-distance} setting, the value of $d$ is known in advance, and the admissible strategy may depend on $d$. In the \emph{unknown-distance} setting, a single admissible sequence $a=(a_i)_{i\geq0}$ must be chosen independently of $d$ and used for all initial distances.
Let $\mathcal A$ denote the set of admissible oblivious self-distance strategies.
For each $d\geq 1$, we define
$$
R_d:=\inf_{a \in \mathcal A} \E[\mathcal T_d(a)]
$$
where the infimum is taken over all admissible oblivious self-distance strategies. 
Thus, $R_d\leq \E[\mathcal T_d(a)]$ for every $a\in\mathcal A$. In particular, $R_d$ is a lower bound on the expected rendezvous time at distance $2d$ of every strategy admissible in the unknown-distance setting.

Both settings are related to versions of symmetric rendezvous on the line studied in the literature. The difference in the present model is that the agents move one edge at a time on $\mathbb Z$, while their decisions are restricted to depend only on the distance from their own starting point. Thus the randomized strategy is specified by a self-distance sequence $a$. In particular, each move depends only on the agent's current distance from her own starting point, and not on the realized history of the walk.


\subsection{Markovian Reformulation}
\label{sec:MarkovianReformulation}

The purpose of this section is to reformulate the rendezvous problem in a Markovian framework. This allows us to express the expected rendezvous time as a hitting time in an absorbing Markov chain, which will later be used for both computation and optimization.

We first define the Markov chain associated with a fixed admissible sequence $(a_i)_{i\geq 0}$ and a fixed initial distance $2d$.
It is convenient to extend the sequence $(a_i)_{i\geq 0}$ to a \emph{signed sequence} $(b_i)_{i\in\mathbb Z}$ by setting
$$
b_i=
\begin{cases}
a_i, & i\geq 0,\\
1-a_{-i}, & i<0.
\end{cases}
$$
Thus $b_i$ represents the probability that an agent whose displacement from her own starting point is $i$ increases this displacement by one in the analysis coordinate system. This is only a notational device for the analysis and does not imply that the agents share a common orientation.

\begin{lemma}
\label{lem:MarkovianReformulation}
Fix $d\geq 1$ and an admissible sequence $(a_i)_{i\geq 0}$. Let
$$
\Omega_d=
\left\{(x,y)\in\mathbb Z^2: x-y\leq 2d,\ x-y\equiv 0 \pmod 2\right\}.
$$
For each $(x,y)\in\Omega_d$, let $S_{x,y}$ be a state, and define
$
\mathcal S_d=\left\{S_{x,y}:(x,y)\in\Omega_d\right\}.
$
The initial state is $S_{0,0}$.
Define
$
\mathcal A_d=
\left\{S_{x,y}\in\mathcal S_d:x-y=2d\right\}.
$
The states in $\mathcal A_d$ are absorbing, that is,
$P\left[S_{x,y},S_{x,y}\right]=1$
for every $S_{x,y}\in\mathcal A_d$.
For every state $S_{x,y}\in\mathcal S_d\setminus\mathcal A_d$, the transition probabilities are
$$
P\left[S_{x,y},S_{u,v}\right]=
\begin{cases}
b_x b_y, & (u,v)=(x+1,y+1),\\
b_x(1-b_y), & (u,v)=(x+1,y-1),\\
(1-b_x)b_y, & (u,v)=(x-1,y+1),\\
(1-b_x)(1-b_y), & (u,v)=(x-1,y-1),\\
0, & \text{otherwise}.
\end{cases}
$$

Then this Markov chain is the chain induced by the distance-$d$ oblivious self-distance symmetric rendezvous problem on the line.
\end{lemma}

\begin{proof}
In the distance-$d$ problem, the two agents start at positions $-d$ and $d$. If their displacements from their respective starting positions are $x$ and $y$, then their actual positions are
$
-d+x
$
and
$
d+y.
$
Thus the configuration is represented by the state $S_{x,y}$.

The two agents meet exactly when
$
-d+x=d+y,
$
which is equivalent to
$
x-y=2d.
$
These are precisely the states in $\mathcal A_d$, and they are made absorbing because the purpose of the chain is to compute the first hitting time of the meeting set.

Before meeting, we have $x-y<2d$. Since both agents move by one unit at every step, the parity of $x$ and $y$ changes simultaneously, so $x-y$ remains even. Together with the absorbing states, this gives the set $\Omega_d$.

Finally, from displacement $x \in \integers$, the first agent moves to $x+1$ with probability $b_x$ and to $x-1$ with probability $1-b_x$. Similarly, from displacement $y$, the second agent moves to $y+1$ with probability $b_y$ and to $y-1$ with probability $1-b_y$. The agents sample their moves independently, so the four transition probabilities are exactly the products listed above. Hence the stated Markov chain is precisely the Markovian representation of the problem.
\end{proof}

The preceding lemma identifies the original rendezvous process with a Markov chain. This chain is generally infinite, and the corresponding system of equations for expected hitting times is therefore infinite as well. In the constructions of Section~\ref{sec:SmallDFrameworks} below, we replace it by finite absorbing Markov chains obtained either by truncating the admissible strategies or by modifying the process at a boundary. 
The expected rendezvous time corresponds to the expected hitting time of the absorbing set $\mathcal A_d$ from the initial state $S_{0,0}$ of Lemma~\ref{lem:MarkovianReformulation}. Since the underlying chain is generally infinite, we next introduce finite modifications that will allow us to compute upper and lower bounds to $R_d$ for small values of $d$.

\section{Oblivious Rendezvous Bounds for Small Values of $d$}
\label{sec:SmallDFrameworks}

This section addresses the known-distance setting, where the initial separation parameter $d$ is given in advance and the admissible strategy may depend on $d$. We develop two complementary finite-state frameworks for this setting. The first yields computable upper bounds through truncated admissible strategies, while the second provides lower bounds through finite relaxations of the original rendezvous process. 
The following standard fact will be used in both cases.

\begin{lemma}
\label{lem:finite-hitting-system}
Let $\mathcal M$ be a finite Markov chain with state space $\mathcal S$ and transition probabilities $P(S,S')$. Let $\mathcal A\subseteq\mathcal S$ be an absorbing set, meaning that
$
P(S,S)=1
$
for every $S\in\mathcal A$. Let $(Z_t)_{t\geq 0}$ denote the chain, and define the hitting time
$
\tau_{\mathcal A}:=\inf\{t\geq 0:Z_t\in\mathcal A\}.
$
For each $S\in\mathcal S$, set
$
E_S:=\mathbb E[\tau_{\mathcal A}\mid Z_0=S].
$
Assume that $E_S<\infty$ for every $S\in\mathcal S$.
Then the values $(E_S)_{S\in\mathcal S}$ are the unique solution of the finite linear system
$$
E_S=
\begin{cases}
0, & \text{if } S\in\mathcal A,\\
1+\sum_{S'\in\mathcal S}P(S,S')E_{S'}, & \text{if } S\in\mathcal S\setminus\mathcal A.
\end{cases}
$$
In particular, if the initial state is $S_{\mathrm{init}}$, then the expected hitting time from the initial state is $E_{S_{\mathrm{init}}}$.
\end{lemma}

\begin{proof}
If $S\in\mathcal A$, then $\tau_{\mathcal A}=0$, and hence $E_S=0$.
Let
$
\mathcal S_0=\mathcal S\setminus\mathcal A.
$
For $S\in\mathcal S_0$, one transition is taken before the process is evaluated again. Since $E_{S'}=0$ for $S'\in\mathcal A$, the Markov property and the law of total expectation give
$$
E_S
=
1+\sum_{S'\in\mathcal S_0}P(S,S')E_{S'},
\qquad S\in\mathcal S_0.
$$

Let $Q$ be the matrix indexed by $\mathcal S_0$ with entries
$
Q(S,S')=P(S,S')
$,
where $S,S'\in\mathcal S_0$.
Writing $E=(E_S)_{S\in\mathcal S_0}$, the system on the nonabsorbing states is
$$
(I-Q)E=\mathbf 1.
$$
Since the chain is finite and $E_S<\infty$ for every $S\in\mathcal S$,
, starting from any non-absorbing state the process reaches an absorbing state with probability $1$. It follows that every eigenvalue of $Q$ has modulus strictly smaller than $1$. Therefore $1$ is not an eigenvalue of $Q$, and hence $I-Q$ is invertible, so the system has a unique solution on $\mathcal S_0$. 
\end{proof}

\subsection{Truncated Chains: An Upper Bounds Framework}
\label{sec:TruncatedChains}

We first consider a finite family of admissible strategies obtained by forbidding agents from moving beyond a prescribed distance from their starting positions.
The associated Markov chains will be used to derive upper bounds to $R_d$ in the known-distance setting, and for small values of $d$. 

\begin{definition}[Level-$n$ Truncated Chain]
\label{def:n-TruncatedChain}
Fix $d\geq 1$, $n\geq d$, and
$
a=(a_1,\ldots,a_{n-1})\in[0,1]^{n-1}.
$
Set
$a_0=\frac12$, and set $a_i=0$ for all $i\geq n$. Let $(b_i)_{i\in\integers}$ be the corresponding signed sequence of Section~\ref{sec:MarkovianReformulation}. 

The \emph{level-$n$ truncated chain} is the Markov chain from Lemma~\ref{lem:MarkovianReformulation}, restricted to the finite state space
$$
\mathcal S^{\mathrm{tr}}_{d,n}
=
\left\{S_{x,y}:-n\leq x,y\leq n,\ x-y\leq 2d,\ x-y\equiv 0 \pmod 2\right\}.
$$
Its absorbing set is
$
\mathcal A^{\mathrm{tr}}_{d,n}
=
\left\{S_{x,y}\in\mathcal S^{\mathrm{tr}}_{d,n}:x-y=2d\right\}.
$
The initial state is $S_{0,0}$, and the transition probabilities are those of Lemma~\ref{lem:MarkovianReformulation}.
\end{definition}
In other words, the level-$n$ truncated chain agrees with the Markov chain of Lemma~\ref{lem:MarkovianReformulation} at all states with self-distances strictly smaller than $n$. The only modification is at self-distance at least $n$: at such a state, the next move is forced toward the starting point by setting $a_i=0$ for all $i\geq n$ (and setting only $a_n=0$ would be equivalent). Thus a coordinate can never leave the interval $[-n,n]$, but after it moves back inside the interval the transition probabilities are again those prescribed by the corresponding self-distance distribution $(a_i)_{i\geq 0}$. 
Let $(Z_t)_{t\geq0}$ denote the level-$n$ truncated chain, and let
$
\tau^{\mathrm{tr}}_{d,n}
=
\inf\left\{t\geq 0:Z_t\in\mathcal A^{\mathrm{tr}}_{d,n}\right\}.
$
We define
$$
U_{d,n}(a)
=
\mathbb E\left[\tau^{\mathrm{tr}}_{d,n}\mid Z_0=S_{0,0}\right],
$$
whenever this expectation is finite. Otherwise, we set $U_{d,n}(a)=\infty$.
The next lemma shows how the chain of Definition~\ref{def:n-TruncatedChain} gives an upper bound to $R_d$. 

\begin{lemma}
\label{lem:UpperBoundFrame}
For every $d\geq 1$, 
every $a=(a_1,\ldots,a_{n-1})\in[0,1]^{n-1}$, 
and every $n\geq d$, we have
$$R_d \leq U_{d,n}(a).$$
\end{lemma}

\begin{proof}
The strategy defined by $a_0=\frac12$ and $a_i=0$ for all $i\geq n$, together with the chosen values of $a_1,\ldots,a_{n-1}$, is an admissible self-distance strategy for the original problem. Since $a_n=0$, no displacement outside the interval $[-n,n]$ can occur. Hence the finite truncated chain agrees exactly with the original rendezvous process under this admissible strategy. Therefore its expected hitting time $U_{d,n}(a)$ from $S_{0,0}$ is the expected rendezvous time of that strategy, which cannot be smaller than the optimal value. 
\end{proof}

By Lemma~\ref{lem:finite-hitting-system}, 
the values $E_{x,y}$ on the states reachable from $S_{0,0}$ are the unique solution of the finite linear system restricted to those states with 
$E_{x,y}=0$, if $S_{x,y}\in\mathcal A^{\mathrm{tr}}_{d,n}$, and
$$
E_{x,y}
=
1+b_xb_yE_{x+1,y+1}
+b_x(1-b_y)E_{x+1,y-1}
+(1-b_x)b_yE_{x-1,y+1}
+(1-b_x)(1-b_y)E_{x-1,y-1}
$$
for every nonabsorbing reachable state $S_{x,y}$.
The upper bound obtained from this parameter vector is
$
U_{d,n}(a)=E_{0,0}.
$

Note, that the Markov chain of Definition~\ref{def:n-TruncatedChain} is indeed finite. 
The number of absorbing states is
$
|\mathcal A^{\mathrm{tr}}_{d,n}|=2n+1-2d
$,
while the number of nonabsorbing variables in the linear system is
$$
N^{\mathrm{tr}}_{d,n}
=
\sum_{q=-n}^{d-1}\left(2n+1-2|q|\right)=n^2+d(2n+1)-d(d-1).
$$

\subsection{Weak-Peek Chains: A Lower Bounds Framework}
\label{sec:PeekChains}

We next introduce a finite aided modification of the Markov chain of Lemma~\ref{lem:MarkovianReformulation}. 
The purpose of
this model is twofold. First, it provides computable lower bounds for the rendezvous
value $R_d$, with the goal of obtaining tight bounds for small values of $d$.
Second, it helps expose the support of optimal oblivious self-distance strategies.

The idea is to let the process evolve according to the original strategy until one of
the agents reaches distance $n$ from her starting point, and then modify the
continuation beyond this level. 
The resulting finite system is optimistic relative to the original process. 
Intuitively, whenever an agent reaches boundary distance at least $n$ from her starting point, it receives momentarily directional information that can only help the agents meet.
Consequently, the expected rendezvous time in the modified process is at most that of
the original process, and the resulting finite system produces a valid lower bound.

Recall that an admissible oblivious self-distance strategy is encoded by a sequence
$
(a_i)_{i\geq 0},
$
where $a_k$ denotes the probability that an agent at distance $k$ from her starting
point moves one step farther away. For some fixed $n\geq d$, the \emph{weak-peek} modification allows the optimization
to decide whether reaching distance $n$ is actually useful. Indeed, if the minimizer
sets some earlier parameter $a_k$ equal to zero, then no agent can ever move beyond
distance $k$. Consequently, the level-$n$ boundary is never reached, and the
weak-peek chain coincides on all reachable states with the corresponding truncated
chain. This is the mechanism by which the weak-peek construction certifies the support of an optimal strategy.

We now introduce the \emph{weak-peek model}. In this model, once an agent reaches distance
$n$ from her starting point, she is allowed to move one step toward the other agent, only for that step.
Note that only local directional information is revealed at the boundary, only to the agent who has reached that boundary, and only at the step in which the boundary is reached. The process then continues according to the modified transition rules.

\begin{definition}[Level-$n$ Weak-Peek Chain]
\label{def:n-WeakPeekChain}
Fix $d\geq 1$, $n\geq d$, and
$
a=(a_1,\ldots,a_{n-1})\in[0,1]^{n-1}.
$
Set $a_0=\frac12$, and let $(b_i)_{i\in\integers}$ be the corresponding signed sequence. The level-$n$ weak-peek chain is the Markov chain with state space
$$
\mathcal S^{\mathrm{wp}}_{d,n}
=
\left\{
S_{i,j}:
-n\leq i\leq n+2d,\ 
-n-2d\leq j\leq n,\ 
i-j\leq 2d,\ 
i-j\equiv 0 \pmod 2
\right\},\footnote{The asymmetric coordinate bounds are chosen so that absorption remains reachable after one coordinate reaches the truncation threshold. Beyond that threshold, motion is forced deterministically toward rendezvous.}
$$
and absorbing set
$$
\mathcal A^{\mathrm{wp}}_{d,n}
=
\left\{
S_{i,j}\in\mathcal S^{\mathrm{wp}}_{d,n}:i-j=2d
\right\}.
$$
The transition probabilities are defined as follows. Each state in
$\mathcal A^{\mathrm{wp}}_{d,n}$ has self loop probability $1$.
For every
$
S_{i,j}\in
\mathcal S^{\mathrm{wp}}_{d,n}
\setminus
\mathcal A^{\mathrm{wp}}_{d,n},
$
the first coordinate moves one step toward rendezvous, namely to $i+1$, with probability $1$ if $|i|\geq n$, and otherwise moves to $i+1$ with probability $b_i$ and to $i-1$ with probability $1-b_i$. 
Similarly, the second coordinate moves one step toward rendezvous, namely to $j-1$, with probability $1$ if $|j|\geq n$, and otherwise moves to $j+1$ with probability $b_j$ and to $j-1$ with probability $1-b_j$. The two coordinates move independently according to these rules. All transition probabilities not specified above are equal to $0$.
\end{definition}
An explanation of the differences between the finite chains of Definitions~\ref{def:n-TruncatedChain},~\ref{def:n-WeakPeekChain} is in place. 
Unlike the level-$n$ truncated chain of Definition~\ref{def:n-TruncatedChain}, reaching distance $n$ does not prevent an agent from subsequently attaining larger self-distances. The only modification is that whenever an agent is currently at distance at least $n$ from her starting point, her next move is forced one step toward rendezvous. After that step, if the agent is again at distance less than $n$, she resumes following the prescribed self-distance distribution of Section~\ref{sec:MarkovianReformulation}. 
Consequently, states with self-distance larger than $n$ may still be reached. However, whenever an agent is at self-distance at least $n$, its next move is forced toward rendezvous. Thus the first coordinate cannot fall below $-n$ and, once it reaches $n$, it moves deterministically upward until rendezvous; similarly, the second coordinate cannot exceed $n$ and, once it reaches $-n$, it moves deterministically until rendezvous. Since the initial separation is $2d$, rendezvous occurs before the first coordinate can exceed $n+2d$ or the second can fall below $-n-2d$. Consequently, the weak-peek chain remains finite.
Let $(Z_t)_{t\geq0}$ denote this chain, and let
$
\tau^{\mathrm{wp}}_{d,n}
=
\inf\left\{t\geq 0:Z_t\in\mathcal A^{\mathrm{wp}}_{d,n}\right\}
$.
We define
$$
L^{\mathrm{wp}}_{d,n}(a)
=
\mathbb E\left[
\tau^{\mathrm{wp}}_{d,n}
\mid
Z_0=S_{0,0}
\right].
$$
The weak-peek modification reveals additional directional information once an agent reaches distance $n$ from her starting point. The next lemma shows that this information can only help the agents, in the sense that the weak-peek chain can always be coupled to stay at least as close to rendezvous as the original process.

\begin{lemma}
\label{lem:weak-peekMonotoneCoupling}
Fix $d\geq 1$, $n\geq d$, and an admissible strategy
$
(a_i)_{i\geq 0}.
$
Let $(X_t,Y_t)$ be the chain of Lemma~\ref{lem:MarkovianReformulation}, and let $(X'_t,Y'_t)$ be the corresponding level-$n$ weak-peek chain of Definition~\ref{def:n-WeakPeekChain}. Let $\mathcal T_d$ and $\mathcal T_d'$ be their respective hitting times, and set
$
\tau=\min\{\mathcal T_d,\mathcal T_d'\}.
$
Then the two processes can be coupled so that, for every $0\leq t\leq\tau$,
$
X'_t\geq X_t
$
and
$
Y'_t\leq Y_t.
$
\end{lemma}

\begin{proof}
We show by induction on $t=0,\ldots, \tau$ that,
$X'_t\geq X_t$. 
The proof of $Y'_t\leq Y_t$ is analogous, with the order reversed. In the equality case we use the same uniform random variable for the second-coordinate update in both chains and use that the weak-peek probability of increasing the second coordinate is at most the corresponding probability in the chain from Lemma~\ref{lem:MarkovianReformulation}. In the case $Y_t-Y'_t\geq 2$, the gap can decrease by at most two in one step, so the order is preserved.

For the base case, we have that $X'_0 = X_0=0$. 
In the inductive step, assume that $t<\tau$ and that $X'_t\geq X_t$. Since both $X_t$ and $X'_t$ have the same parity, there is an integer $m\geq 0$ such that
$
X'_t-X_t=2m.
$
If $m\geq 1$, then $X'_t-X_t\geq 2$. Since each coordinate changes by one unit at the next step,
$$
X'_{t+1}-X_{t+1}
\geq
(X'_t-1)-(X_t+1)
=
X'_t-X_t-2
\geq 0.
$$
Hence
$
X'_{t+1}\geq X_{t+1}
$. 
It remains to consider the case $m=0$, that is the case $X'_t-X_t=0$
So, assume that $X'_t=X_t=k$, for some $k\in \integers$. 

Let $(b_i)_{i\in\integers}$ be the signed move probabilities associated with $(a_i)_{i\geq0}$, as in Lemma~\ref{lem:MarkovianReformulation}, and recall that in the process from Lemma~\ref{lem:MarkovianReformulation}, a coordinate currently equal to $k$ increases to $k+1$ with probability $b_k$, and decreases to $k-1$ with probability $1-b_k$.
For the weak-peek chain, let $p_X(k)$ denote the probability that the first coordinate increases from $k$ to $k+1$ in one step, and let $p_Y(k)$ denote the corresponding probability for the second coordinate. By Definition~\ref{def:n-WeakPeekChain},
$$
p_X(k)=
\begin{cases}
b_k, & |k|<n,\\
1, & |k|\geq n,
\end{cases}
$$
Thus, for every $k\in\integers$, we have $p_X(k)\geq b_k$.
Let now $U$ be uniform on $[0,1]$, and set
$$
X_{t+1}-X_t=
\begin{cases}
+1, & U\leq b_k,\\
-1, & U>b_k,
\end{cases}
\qquad
X_{t+1}'-X_t'=
\begin{cases}
+1, & U\leq p_X(k),\\
-1, & U>p_X(k).
\end{cases}
$$
Since $p_X(k)\geq b_k$, we have $X_{t+1}'-X_t' \geq X_{t+1}-X_t$, that is 
$
X'_{t+1} - X_{t+1} \geq X_t' -X_t = 0.
$
This completes the inductive step. 
\end{proof}

The monotone coupling of Lemma~\ref{lem:weak-peekMonotoneCoupling} implies that the weak-peek chain reaches rendezvous no later than the corresponding original process. Consequently, optimizing over weak-peek chains yields a valid lower bound for the original rendezvous problem. This is formalized in the next lemma. 

\begin{lemma}
\label{lem:WeakPeekLowerBound}
For every $d\geq 1$ and every $n\geq d$, we have 
$
\inf_{a\in[0,1]^{n-1}} L^{\mathrm{wp}}_{d,n}(a)
\leq
R_d.
$
\end{lemma}

\begin{proof}
Fix an arbitrary admissible strategy
$
\alpha=(\alpha_i)_{i\geq 0}
$
with $\alpha_0=1/2$, and let
$
a=(\alpha_1,\ldots,\alpha_{n-1}).
$
Let $\mathcal T_d$ and $\mathcal T_d'$ be as in Lemma~\ref{lem:weak-peekMonotoneCoupling}. By that lemma, the two processes can be coupled so that, up to
$
\tau=\min\{\mathcal T_d,\mathcal T_d'\},
$
we have
$
X'_t\geq X_t
$
and
$
Y'_t\leq Y_t.
$
Hence
$
X'_t-Y'_t\geq X_t-Y_t
$, 
for all $0\leq t\leq \tau$.
We claim that $\mathcal T_d'\leq \mathcal T_d$. Indeed, suppose for contradiction that
$
\mathcal T_d'>\mathcal T_d.
$
Then
$
\tau=\min\{\mathcal T_d,\mathcal T_d'\}=\mathcal T_d.
$
Since $\mathcal T_d$ is the hitting time of the absorbing set of the process from Lemma~\ref{lem:MarkovianReformulation}, we have
$
X_{\mathcal T_d}-Y_{\mathcal T_d}=2d.
$
Since the inequalities from Lemma~\ref{lem:weak-peekMonotoneCoupling} hold at time $\mathcal T_d$, we have
$$
X'_{\mathcal T_d}-Y'_{\mathcal T_d}
\geq
X_{\mathcal T_d}-Y_{\mathcal T_d}
=
2d.
$$
Thus the weak-peek chain has already reached its absorbing set by time $\mathcal T_d$, contradicting the assumption that
$
\mathcal T_d'>\mathcal T_d.
$
Taking expectations, we conclude that 
$
L^{\mathrm{wp}}_{d,n}(a)
=
\mathbb E[\mathcal T_d']
\leq
\mathbb E_{\alpha}[\mathcal T_d].
$
Subsequently, taking the infimum over all admissible $\alpha$ gives
$$
\inf_{a\in[0,1]^{n-1}}L^{\mathrm{wp}}_{d,n}(a)
\leq
\inf_{\alpha \in \mathcal A}\mathbb E_{\alpha}[\mathcal T_d]
=
R_d.
$$
\end{proof}

\subsection{Established Bounds to $R_d$ for Small Values of $d$}
\label{sec:AllBoundsForSmallD}

In this section, we apply the upper-bound framework of Section~\ref{sec:TruncatedChains} together with the weak-peek framework of Section~\ref{sec:PeekChains} to small values of $d$. For $d=1$, this leads to a complete proof of optimality. For larger values of $d$, we use extensive numerical optimization to search for the structural condition required by the weak-peek certificate. The resulting truncated strategies give rigorous upper bounds, while their optimality would follow if the numerically obtained weak-peek solutions were global minimizers, a condition that we do not prove.

The certificate we use is a sandwich argument. The optimized weak-peek model is a lower bound on $R_d$, while every truncated strategy gives an upper bound. If a global weak-peek minimizer has $a_k=0$ for some $k<n$, then the weak-peek boundary is never reached and the two finite chains coincide; hence the lower and upper bounds are equal.

The key ingredient is the following structural lemma. It shows that whenever a global minimizer of the weak-peek problem has a sufficiently early vanishing parameter, the corresponding weak-peek chain coincides on all reachable states with a truncated admissible strategy, effectively resulting in an optimal solution to the oblivious rendezvous problem, with initial displacement $2d$.

\begin{lemma}
\label{lem:WeakPeekCertificate}
Fix $d\geq 1$ and $n>d$. Let
$
a=(a_1,\ldots,a_{n-1})\in[0,1]^{n-1}
$
be a global minimizer of the level-$n$ weak-peek problem. Suppose that
$
a_k=0
$
for some $d\leq k<n$.
Then
$
R_d
=
U_{d,k}(a_1,\ldots,a_{k-1})
$.
In particular, the corresponding level-$k$ truncated strategy is optimal for the distance-$d$ oblivious self-distance rendezvous problem.
\end{lemma}

\begin{proof}
Since $a_k=0$, an agent at distance $k$ from her starting point moves toward her starting point with probability $1$. Hence no agent can move from distance $k$ to distance $k+1$. Starting from $S_{0,0}$, all reachable states satisfy
$
|x|,|y|\leq k.
$
Since $k<n$, the weak-peek rule is never used on any reachable state.

It follows that, under $a$, the level-$n$ weak-peek chain agrees on all reachable states with the level-$k$ truncated chain defined by $(a_1,\ldots,a_{k-1})$.
Therefore
$
L^{\mathrm{wp}}_{d,n}(a)
=
U_{d,k}(a_1,\ldots,a_{k-1}).
$
By Lemma~\ref{lem:WeakPeekLowerBound}, and since $a$ minimizes the level-$n$ weak-peek value,
$$
L^{\mathrm{wp}}_{d,n}(a)
=
\inf_{a\in[0,1]^{n-1}}L^{\mathrm{wp}}_{d,n}(a)
\leq
R_d.
$$
On the other hand, the level-$k$ truncated chain defines an admissible self-distance strategy for the original problem. Hence, by the truncated-chain upper bound,
$
R_d
\leq
U_{d,k}(a_1,\ldots,a_{k-1}).
$
Combining the relations gives
$$
U_{d,k}(a_1,\ldots,a_{k-1})
=
L^{\mathrm{wp}}_{d,n}(a)
\leq
R_d
\leq
U_{d,k}(a_1,\ldots,a_{k-1}),
$$
and equality follows.
\end{proof}

\subsubsection{The Optimal Oblivious Expected Rendezvous Time for $d=1$}
\label{sec:UpperBoundsSmallD1}


In this section we prove the following result.

\begin{theorem}
\label{thm:optR1}
We have
$
R_1\approx 6.4670988308.
$
Moreover, $R_1$ is attained by a level-$2$ truncated chain with parameter
$
a_1^\ast\approx 0.1992570042,
$
where $a_1^\ast$ is the unique root in $(0,1)$ of
$
-12+52x+39x^2+12x^3-3x^4.
$
Thus the exact value is
$
R_1=U_{1,2}(a_1^\ast).
$
\end{theorem}

Applying Lemma~\ref{lem:finite-hitting-system} to the level-$3$ weak-peek chain of Definition~\ref{def:n-WeakPeekChain} yields a finite linear system for the expected absorption times $E_{i,j}$. Solving this system symbolically gives, whenever $D(a_1,a_2)\neq 0$,
$
E_{0,0}(a_1,a_2)=\tfrac{N(a_1,a_2)}{D(a_1,a_2)}.
$
The polynomials $N$ and $D$ are given in Table~\ref{tab:d1ND}. The left matrix corresponds to $N(a_1,a_2)$ and the right matrix corresponds to $D(a_1,a_2)$. In both matrices, the entry in row $i$ and column $j$ is the coefficient of $a_1^i a_2^j$.

\begin{table}[!h]
\centering

\begin{minipage}{0.48\textwidth}
\centering
\[
\begin{array}{c|rrrrrr}
& a_2^0 & a_2^1 & a_2^2 & a_2^3 & a_2^4 & a_2^5 \\ \hline
a_1^0 & 14 & -14 & -14 & 14 & 0 & 0 \\
a_1^1 & -5 & -7 & 68 & -42 & -21 & 7 \\
a_1^2 & -11 & 52 & -82 & -2 & 79 & -28 \\
a_1^3 & 3 & -26 & 4 & 88 & -115 & 42 \\
a_1^4 & -1 & -2 & 32 & -78 & 77 & -28 \\
a_1^5 & 0 & 1 & -8 & 20 & -20 & 7
\end{array}
\]

{\small (a) $N(a_1,a_2)$}

\end{minipage}
\hfill
\begin{minipage}{0.48\textwidth}
\centering
\[
\begin{array}{c|rrrrrr}
& a_2^0 & a_2^1 & a_2^2 & a_2^3 & a_2^4 & a_2^5 \\ \hline
a_1^0 & 2 & -2 & -2 & 2 & 0 & 0 \\
a_1^1 & 1 & -5 & 8 & -2 & -3 & 1 \\
a_1^2 & -7 & 14 & 6 & -18 & 7 & -2 \\
a_1^3 & 3 & 6 & -40 & 30 & 3 & -2 \\
a_1^4 & 1 & -16 & 30 & -2 & -21 & 8 \\
a_1^5 & 0 & 3 & 0 & -16 & 20 & -7 \\
a_1^6 & 0 & 0 & -2 & 6 & -6 & 2
\end{array}
\]

{\small (b) $D(a_1,a_2)$}

\end{minipage}

\caption{Coefficient matrices for the numerator $N(a_1,a_2)$ and denominator $D(a_1,a_2)$ of $E_{0,0}(a_1,a_2)$. Rows correspond to powers of $a_1$ and columns correspond to powers of $a_2$.}
\label{tab:d1ND}
\end{table}

By Lemma~\ref{lem:WeakPeekCertificate}, it is enough to show that the minimum of $E_{0,0}(a_1,a_2)$ over $[0,1]^2$ is attained at a point with $a_2=0$. Indeed, in that case the level-$3$ weak-peek certificate proves optimality of the corresponding level-$2$ truncated chain. Thus Theorem~\ref{thm:optR1} follows from the next lemma.

\begin{lemma}
\label{lem:d1level3Minimizer}
The function $E_{0,0}(a_1,a_2)$ is minimized over $[0,1]^2$ at
$
(a_1^\ast,0),
$
where $a_1^\ast \approx 0.199257$ is the unique root in $(0,1)$ of
$
-12+52x+39x^2+12x^3-3x^4.
$
Moreover,
$
E_{0,0}(a_1^\ast,0)\approx 6.467098.
$
\end{lemma}

\begin{proof}
First, we find the optimizer of
$$
E_{0,0}(x,0)
=
\frac{x^3-2 x^2+9 x+14}{(x-1) \left(x^2+5 x+2\right)}
$$
Differentiating gives
$$
\frac{\d{}}{\d x}E_{0,0}(x,0)
=
\frac{-12+52x+39x^2+12x^3-3x^4}
{\tfrac12 (x-1)^2(x^4+10x^3+29x^2+20x+4)}.
$$
The denominator is strictly positive for every $x\in(0,1)$. We now show that the numerator has a unique root $a_1^\ast$ in $(0,1)$. Let
$
p(x)=-12+52x+39x^2+12x^3-3x^4.
$
Then
$$
p'(x)
=
52+78x+36x^2-12x^3
\geq
52+78x+36x^2-12
>
0
$$
for every $x\in[0,1]$. Hence $p$ is strictly increasing on $[0,1]$. Since
$
p(0)=-12<0,
$
and 
$
p(1)=88>0,
$
there is a unique root $x^\ast\in(0,1)$. 
It follows that 
$
\frac{\d{}}{\d x}E_{0,0}(x,0)
$
is negative on $(0,x^\ast)$ and positive on $(x^\ast,1)$, and so $E_{0,0}(x,0)$ is minimized at $x^\ast$, which admits a closed formula, but here we only provide the numerical estimate 
$x^\ast \approx 0.199257$.

Next we observe that 
$E_{0,0}(0,y)=7$. 
Moreover, if $x=1$, or if $y=1$ and $x>0$, then the expected absorption time is infinite. Indeed, in either case the chain has a positive-probability event on which the agents move forever without reaching the absorbing set. Since
$
E_{0,0}(x^\ast,0)\approx 6.467098<7,
$
the only possible boundary minimizer is $(x^\ast,0)$.

It remains to exclude minimizers in the interior, but we first justify that a global minimum is attained. If $\tau$ denotes the absorption time, then
$
E_{0,0}(x,y)=\sup_{m\geq 0}\sum_{t=0}^{m}\Pr_{x,y}(\tau>t).
$
Each finite sum is continuous in $(x,y)$, since the transition probabilities are continuous in the parameters. Hence $E_{0,0}$ attains its minimum. Thus the lemma follows if we prove that $E_{0,0}$ has no minimizer in $(0,1)^2$. As we describe next, we verify this with symbolic calculations using Mathematica.

Since
$
E_{0,0}(x,y)=N(x,y)/D(x,y)
$
is differentiable at every point where $D(x,y)\neq 0$, any minimizer in $(0,1)^2$ must satisfy the first order KKT conditions 
$$
\frac{\partial}{\partial x}E_{0,0}(x,y)=0,
\qquad
\frac{\partial}{\partial y}E_{0,0}(x,y)=0.
$$
Writing
$
G_x(x,y):=(\partial_xD)N-(\partial_xN)D
$,
and
$
G_y(x,y):=(\partial_yD)N-(\partial_yN)D
$,
the above conditions are equivalent to
$$
G_x(x,y)=0,
\qquad
G_y(x,y)=0,
\qquad
D(x,y)\neq 0.
$$
Since
$
D(x,y)=(x-1)(y-1)\widetilde D(x,y)
$
and $(x-1)(y-1)\neq 0$ on $(0,1)^2$, the condition $D(x,y)\neq 0$ is equivalent to
$
\widetilde D(x,y)\neq 0.
$
Therefore it suffices to determine whether the semi-algebraic set
$$
\left\{
(x,y)\in\mathbb R^2:
G_x(x,y)=0,\,
G_y(x,y)=0,\,
\widetilde D(x,y)\neq 0,\,
0<x<1,\,
0<y<1
\right\}
$$
is empty.
We verified this using Mathematica's \texttt{Reduce} command based on
cylindrical algebraic decomposition~\cite{Collins-1975-Cylindrical};
see Appendix~\ref{sec:MathematicaCode}.
The computation is symbolic rather than numerical and returns \texttt{False}. 
Hence the corresponding semi-algebraic set is empty, and therefore $E_{0,0}$ has no critical point in $(0,1)^2$.
\end{proof}

\subsubsection{Empirical Rendezvous Upper Bounds for $d=1,\ldots,6$} \label{sec:NumericalComputations} 


For $d=1$, the weak-peek minimization problem was solved exactly in the previous section, yielding a rigorous optimality proof. For larger values of $d$, the corresponding weak-peek objective functions become too complicated for a direct symbolic analysis. We therefore turn to numerical computations. The goal remains the same: to identify weak-peek minimizers whose first vanishing parameter may certify optimality of a truncated strategy through Lemma~\ref{lem:WeakPeekCertificate}. While we do not prove global optimality of the computed weak-peek minimizers, the resulting truncated strategies still provide rigorous computable upper bounds for $R_d$. If the corresponding weak-peek minimizers are globally optimal, then Lemma~\ref{lem:WeakPeekCertificate} implies that these upper bounds are exact.

The computations were carried out in Python, using NumPy and SciPy.  For each
fixed pair $(d,n)$, we considered the level-$n$ weak-peek chain from
Definition~\ref{def:n-WeakPeekChain}.  For each fixed vector
$(a_1,\ldots,a_{n-1})\in[0,1]^{n-1}$, with $a_0=1/2$, we constructed the finite
linear system from Lemma~\ref{lem:finite-hitting-system} corresponding to the
expected hitting times of the weak-peek chain.  The value
$
L^{\mathrm{wp}}_{d,n}(a)
$
was then evaluated numerically from the corresponding sparse linear system using
\texttt{scipy.sparse.linalg.spsolve}.
The optimization over $[0,1]^{n-1}$ was performed using the L-BFGS-B routine~\cite{byrd1995limited} from
\texttt{scipy.optimize.minimize}, with box constraints $0\leq a_i\leq 1$.
For each pair $(d,n)$, the optimization was initialized from a collection of deterministic and randomized starting points. The deterministic points included profiles obtained from previous computations and several simple benchmark profiles. The remaining starting points were generated from a mixture of random vectors in $[0,1]^{n-1}$ and perturbations of the previously computed profiles.

For $d=1,\ldots,5$, we used 20 independent random initial points for each reported
pair $(d,n)$, and report the smallest value obtained among these runs.  These
computations were stable: the same candidate truncation indices and essentially the
same objective values were repeatedly obtained.
For $d\geq 6$, the computations become substantially harder. For the reported case
$d=6$, we ran the weak-peek optimization for
$
n=2d,2d+1,\ldots,10d.
$
For each pair $(d,n)$, we used 200 independent random initial points.  For each
fixed $d$, we recorded the smallest value found, and also tracked near-best
certificate candidates.  A parameter was treated as numerically zero when
$
a_k\leq 10^{-7}.
$
Among values within $10^{-8}$ of the best value found, we selected the candidate
with the earliest reported vanishing parameter.  We then looked for the same
near-best vanishing parameter to appear repeatedly for consecutive values of $n$.
Once such a candidate was detected, the corresponding value of $n$ was rerun with
500 independent random initial points as an additional stability check.

For every reported computation, if $AE=b$ denotes the sparse linear system arising
from Lemma~\ref{lem:finite-hitting-system} for the parameter vector being evaluated,
then the residual norm satisfied
$
\|AE-b\|_\infty\leq 10^{-7}
$
for the computations with $d=1,\ldots,5$, and
$
\|AE-b\|_\infty\leq 10^{-9}
$
for the computations with $d= 6$.
By Lemma~\ref{lem:WeakPeekLowerBound},
$
\inf_{a\in[0,1]^{n-1}}L^{\mathrm{wp}}_{d,n}(a)\leq R_d.
$
Thus the level-$n$ weak-peek chain gives a computable finite lower-bound model for
the rendezvous value.

For each reported pair $(d,n)$, we list the smallest weak-peek value found by the optimization procedure together with the corresponding parameter vector. Whenever the reported solution satisfies
$
a_k=0
$
for some $d\leq k<n$, we also report the smallest such index $k$. By Lemma~\ref{lem:WeakPeekCertificate}, if the reported solution is a global minimizer of the level-$n$ weak-peek problem, then
$$
R_d
=
U_{d,k}(a_1,\ldots,a_{k-1}),
$$
and the corresponding level-$k$ truncated strategy is optimal.

By Lemma~\ref{lem:finite-hitting-system}, each evaluation of
$
L^{\mathrm{wp}}_{d,n}
$
requires solving a sparse linear system with
$
(n+d)^2
$
unknowns. Thus the computation may be viewed as a nonlinear optimization problem with
$
n-1
$
box-constrained variables together with
$
(n+d)^2
$
auxiliary variables linked by
$
(n+d)^2
$
linear constraints.

Lemma~\ref{lem:WeakPeekCertificate} shows that if a minimizer of the level-$n$
weak-peek problem satisfies $a_k=0$ for some $d\leq k<n$, then the corresponding
level-$k$ truncated strategy is optimal.
In computations, however, the weak-peek boundary is not neutral.  It gives the
agents an artificial advantage once distance $n$ is reached, because the boundary
rule reveals the direction toward rendezvous.  Therefore, when $n$ is too close to
the true truncation level, the optimizer may be incentivized to push probability
mass toward reaching the weak-peek boundary, rather than setting an earlier
parameter equal to zero.  In that case, the numerical minimizer may keep $a_k>0$
even though the eventual truncated strategy is the structure we want to detect.

For this reason, it can be useful to run the weak-peek optimization at levels $n$
larger than the expected truncation index.  Moving the aided boundary farther away
makes it less beneficial to exploit the weak-peek information, and the computation
is then more likely to reveal the earlier vanishing parameter. Thus larger values
of $n$ are not needed because of the certificate itself; they are useful because
they make the numerical search less biased toward strategies that use the
artificial help available at the weak-peek boundary.
We emphasize that these computations do not establish global optimality. The resulting truncated-strategy upper bounds are rigorous, while optimality of the corresponding strategies remains a conjecture.

Table~\ref{tab:weakpeek-summary} reports the main numerical results obtained from
the weak-peek optimization for the tested values of $d$.  For each value of $d$,
we report the selected weak-peek computation together with the smallest value of
$E_{0,0}$ found.  The reported index $k$ denotes the first parameter numerically
reported as zero.  Under the certificate property proved earlier, if the reported
solution is in fact a global minimizer of the corresponding weak-peek problem, then
the induced level-$k$ truncated strategy is optimal, and the reported value gives
matching upper and lower bounds for $R_d$.
The corresponding optimized parameter values are listed in
Table~\ref{tab:weakpeek-alphas}, in Appendix~\ref{sec:CandOptSupport}.

\begin{table}[!h]
\centering
\begin{tabular}{c c c c c c c} 
\hline 
& $d=1$ & $d=2$ & $d=3$ & $d=4$ & $d=5$ & $d=6$ \\ 
\hline $n$ & 3 & 6 & 10 & 14 & 18 & 26 \\ 
$k$ & 2 & 5 & 9 & 13 & 17 & 21 \\ 
$E_{0,0}$ & 6.467098 & 30.083326 & 69.418625 & 124.447569 & 195.190221 & 281.650397 \\ 
$E_{0,0}/d^2$ & 6.467098 & 7.520831 & 7.713180 & 7.777973 & 7.807608 & 7.823622 \\ \hline \end{tabular}
\caption{Numerically computed weak-peek candidates for the tested values of $d$.
For each distance $d$, the table reports the weak-peek level $n$, the first numerically vanishing index $k$, the corresponding rigorous truncated-strategy upper bound $E_{0,0}$, and $E_{0,0}/d^2$.}
\label{tab:weakpeek-summary}
\end{table}

\section{Oblivious Rendezvous Bounds for Asymptotic Values of $d$}
\label{sec:AsymptoticBounds}
In contrast to the finite-state constructions (known-distance setting) of Section~\ref{sec:SmallDFrameworks}, the lower bound developed here applies to both the known-distance and unknown-distance settings.
It is universal, as it holds for every admissible strategy and every value of $d$, including strategies that are allowed to depend on $d$ in the known-distance setting.
For upper bounds, we give a universal strategy independent of $d$ for the unknown-distance setting, and a separate $d$-dependent finite-support strategy yielding a quadratic bound in the known-distance setting.

\subsection{Electrical Representation of the Induced Chains}
\label{sec:ElectricalRepresentationInducedChains}

Oblivious self-distance strategies induce Markov chains whose states record the current distance from the origin, or more generally the relative separation between the two agents. In both the lower and upper bound arguments below, the analysis reduces to a random walk on a path. For the lower bound, this arises by tracking the repeated returns to the origin that must occur after unsuccessful rendezvous attempts. For the upper bound, we restrict attention to a specific family of configurations leading to rendezvous, thereby reducing the analysis again to a one-dimensional path in the two-agent state space. So at a high level, in both cases the relevant quantity is a hitting time of a random walk on $\{0,1,\dots\}$ with transitions only between adjacent states. Such \emph{birth-death} chains admit a useful electrical interpretation, allowing hitting times to be controlled through \emph{effective resistance estimates} (for a general treatment, see~\cite{LevinPeresWilmer}). We briefly recall this correspondence, which will be invoked repeatedly in what follows.

More precisely, consider a birth-death chain on $\{0,1,\dots\}$ with transition probabilities
$
P(i,i+1)=p_i
$
and
$
P(i,i-1)=q_i.
$
Whenever there exists a positive measure $\pi$ satisfying
$
\pi(i)p_i=\pi(i+1)q_{i+1},
$
the chain is called \emph{reversible} with respect to $\pi$.
If $p_i,q_{i+1}>0$ for every $i$, then the chain is reversible. Indeed, one may choose $\pi_0=1$ and define recursively
$
\pi_{i+1}
=
\pi_i
\tfrac{p_i}{q_{i+1}},
$
for all $i\geq 0$.

When the chain is reversible, one may define edge \emph{conductances}
$
c(i,i+1)=\pi(i)p_i,
$
which induce an electrical network on the same path, and with \emph{total conductance} $\sum_e c(e)$.

For an electrical network with edge conductances $c(e)$, the \emph{resistance} of an edge $e$ is $1/c(e)$. The \emph{effective resistance} between two states $a,b$ is the electrical resistance measured between $a$ and $b$ in the resulting network. Since the network is a path, this is given by
$$
R(a\leftrightarrow b)
=
\sum_{e\in P(a,b)}
\frac{1}{c(e)},
$$
where $P(a,b)$ denotes the unique path from $a$ to $b$.
The significance of the effective resistance comes from the fact that, for reversible Markov chains, hitting and commute times can be expressed in terms of resistance and total conductance. Thus probabilistic questions about the random walk may be reduced to deterministic calculations in the associated electrical network.

The next lemma will be applied to reversible birth-death chains on a path in the lower-bound argument of Section~\ref{sec:QuadraticLowerBounds}. We state it in its full generality, as the finite commute-time identity (Proposition~10.6 of~\cite{LevinPeresWilmer}), because the upper-bound argument in Section~\ref{sec:QuadraticUpperBound} will require an extension to infinite networks.

\begin{lemma}
\label{lem:ElectricalCommuteIdentity}
Let $X_t$ be a finite reversible Markov chain, and let $c(e)$ be the edge conductances induced by a reversible measure. Denote by 
$c_G$
the total conductance of the network, and for any two states $a,b$, let
$
R(a\leftrightarrow b)
$
denote the effective resistance between $a$ and $b$.
Then
$$
\E[T_{a\to b}]
+
\E[T_{b\to a}]
=
2 c_G ~ R(a\leftrightarrow b),
$$
where 
$T_{x\to y}
$
denotes the hitting time of state $y$ when the chain is started from state $x$.
\end{lemma}

The lower-bound argument in Section~\ref{sec:QuadraticLowerBounds} applies Lemma~\ref{lem:ElectricalCommuteIdentity} only to finite birth-death chains. In the upper-bound argument of Section~\ref{sec:QuadraticUpperBound}, the relevant network is infinite and the target is a set of meeting states. We will therefore use the following extension, obtained by approximating the network with finite networks.

\begin{lemma}
\label{lem:InfiniteResistanceHittingBound}
Let $X_t$ be a countable reversible Markov chain with state space $\mathcal S$ whose induced electrical network has finite total conductance $c_G$. Let $A$ be a nonempty set of states. Define $R(a\leftrightarrow A)$ as the supremum, over all finite sets $B\subseteq\mathcal S\setminus A$ containing $a$, of the effective resistance from $a$ to the vertex obtained by merging $A\cup(\mathcal S\setminus B)$ into one vertex.
Then
$$
\E[T_{a\to A}]
\leq
2c_G\,R(a\leftrightarrow A).
$$
\end{lemma}

\begin{proof}
If $R(a\leftrightarrow A)=\infty$, there is nothing to prove. If $a\in A$, then both sides are zero. Thus we may assume that $R(a\leftrightarrow A)<\infty$ and $a\notin A$.

Let now $\mathcal S$ denote the state space of the reversible chain, and let
$
B_1\subseteq B_2\subseteq\cdots
$
be finite subsets of $\mathcal S\setminus A$ such that $a\in B_m$ for every $m$,
$
\bigcup_m B_m=\mathcal S\setminus A,
$
and, after merging
$
A_m^*=A\cup(\mathcal S\setminus B_m)
$
into one vertex, the resulting finite network contains a path from $a$ to that vertex.

Then $\mathcal S=B_m\cup A_m^*$, and $B_m, A_m^*$ are disjoint. 
For each $m$, form a finite network by merging all states of $A_m^*$ into a single vertex, denoted by $A_m$, and deleting any self-loops created at $A_m$. Let $c_G^{(m)}$ denote the total conductance of this finite network, and let
$
R_m(a\leftrightarrow A_m)
$
denote the corresponding effective resistance.

By construction, $c_G^{(m)}$ is the sum of the conductances of edges with at least one endpoint in $B_m$, after parallel edges to $A_m$ have been merged and self-loops at $A_m$ have been deleted. Hence
$
c_G^{(m)}\leq c_G.
$
By the definition of $R(a\leftrightarrow A)$,
$
R_m(a\leftrightarrow A_m)\leq R(a\leftrightarrow A).
$

Applying now Lemma~\ref{lem:ElectricalCommuteIdentity} (Proposition~10.6 of~\cite{LevinPeresWilmer}) to this finite reversible network gives
$$
\E[T_{a\to A_m}^{(m)}]
=
2c_G^{(m)}R_m(a\leftrightarrow A_m) - 
\E[T_{A_m\to a}^{(m)}]
\leq 2c_G^{(m)}R_m(a\leftrightarrow A_m).
$$

The finite network defines a reversible Markov chain through its conductances. For every $x,y\in B_m$, its transition probability from $x$ to $y$ is the same as in the original chain, while its transition probability from $x\in B_m$ to $A_m$ is the total original probability of moving from $x$ to $A_m^*$. Therefore, the hitting time of $A_m$ in the finite merged chain has the same distribution as
$$
\tau_m=\inf\{t\geq 0:X_t\in A_m^*\}
$$
for the original chain started from $a$. Therefore,
$
\E[\tau_m]
=
\E[T_{a\to A_m}^{(m)}].
$

Since
$
A_m^*\supseteq A_{m+1}^*
$
and
$
\bigcap_m A_m^*=A,
$
we have
$
\tau_m\leq \tau_{m+1}
$
for every $m$, and
$
\tau_m \rightarrow T_{a\to A},
$
from below. 
By the Monotone Convergence Theorem, it follows that 
$$
\E[T_{a\to A}]
=
\lim_{m\to\infty} \E[\tau_m]
= 
\lim_{m\to\infty}	 \E[T_{a\to A_m}^{(m)}]
\leq
2c_G R(a\leftrightarrow A).
$$
\end{proof}

In the arguments below, we will identify suitable measures $\pi$ for the induced chains arising from oblivious self-distance strategies, and then apply Lemma~\ref{lem:ElectricalCommuteIdentity} to estimate the corresponding hitting times.

\subsection{Lower Bound for Oblivious Strategies (Known-Distance Setting)}
\label{sec:QuadraticLowerBounds}

In this section we establish a quadratic lower bound that holds uniformly over all admissible strategies, for every value of $ d\geq 1$, even when the strategy itself is allowed to depend on $d$.
\begin{theorem}
\label{thm:QuadraticLowerBound}
For every $d\geq 1$, we have $R_d = \Omega\left(d^2\right)$.
\end{theorem}
Recall that
$
R_d=\inf_{a \in \mathcal A} \E[\mathcal T_d],
$
where $\mathcal T_d=\mathcal T_d(a)$ for an admissible sequence $(a_i)_{i\geq0}$ with $a_0=1/2$.
Therefore, it suffices to prove the existence of a universal constant $c>0$ such that
$
\E[\mathcal T_d]\geq c d^2
$
for every admissible sequence $(a_i)_{i\geq0}$.
All claims in Section~\ref{sec:QuadraticLowerBounds} are stated for an arbitrary fixed
$d\geq 1$ and an arbitrary fixed admissible sequence $(a_i)_{i\geq0}$. Accordingly,
all probabilities and expectations in this section are taken with respect to that fixed
strategy $(a_i)_{i\geq0}$ and that fixed value of $d$.
Next, we fix $T=\frac12d^2$.\footnote{The choice $T=\frac12 d^2$ is made for simplicity. Optimizing over $T=\theta d^2$ with $\theta\in(0,1)$ slightly improves the numerical constant of Theorem~\ref{thm:QuadraticLowerBound} in the final lower bound, currently equal to $1/512$, but the purpose here is to establish the quadratic order of growth.
Moreover, the proof is based on a sequence of relaxation arguments,
and is therefore not expected to yield a sharp quadratic constant.}
Let $(X_t)$ and $(Y_t)$ denote the displacement processes of the left and right agents, respectively. Define also 
$$
M_{X,T}^+=\max_{0\leq s\leq T}X_s,
\qquad
M_{Y,T}^-=\max_{0\leq s\leq T}(-Y_s).
$$
If rendezvous occurs by time $T$, then for some $s\leq T$,
$
X_s-Y_s=2d
$, and therefore, 
$$
2d=X_s-Y_s\leq M_{X,T}^+ + M_{Y,T}^-.
$$
It follows that, if
$
M_{X,T}^+ + M_{Y,T}^- < 2d,
$
then rendezvous does not occur by time $T$.
Consequently, it is enough to show that this event has probability bounded below by a positive constant.

To obtain such a bound, we study the motion of a single agent. Let $(X_t)_{t\geq 0}$ denote the displacement process of one agent, started from $X_0=0$, and define
$$
\tau_d^+=\inf\{t\geq 0:X_t=d\}.
$$
We compare the signed process $X_t$ with the absolute-value process
$
R_t=|X_t|.
$
Each excursion of $R_t$ away from $0$ corresponds either to a ``positive'' excursion or to a ``negative'' excursion of $X_t$, with the sign (right or left, respectively) chosen independently and uniformly. Hence, after an excursion of $R_t$ reaches level $d$, the probability that the corresponding excursion of $X_t$ reaches $+d$ is exactly $1/2$. 

Hence, conditional on exactly one excursion of $R_t$ reaching level $d$ by time $T$,
the event $\{\tau_d^+\leq T\}$ has probability exactly $1/2$.
Any additional contribution to $\mathbb P \left[ \tau_d^+\leq T \right]$
must therefore come from trajectories in which at least two excursions of $R_t$ reach level $d$ by time $T$.
Two such excursions force the absolute-value chain, after first reaching $d$, to return to $0$ and then reach $d$ again. Thus the key point is to control the probability of completing the commute
$
d\to0\to d
$
too quickly. We formalize this using a truncated version of the absolute-value chain.

We now begin the formal argument. For the displacement process $(X_t)_{t\geq 0}$ of one agent, started from $X_0=0$, let
$
R_t=|X_t|.
$
Let $(\widetilde R_t)_{t\geq 0}$ be the Markov chain on $\{0,1,\ldots,d\}$ obtained from $(R_t)$ by replacing the transition from $d$ to $d+1$ with a deterministic move from $d$ to $d-1$. We call $(\widetilde R_t)$ the reflected chain.
Let $\sigma_d$ be the time for $(\widetilde R_t)$, started from $d$, to hit $0$ and then return to $d$.
Our main technical lemma is the following fast commute obstruction.

\begin{lemma}
\label{lem:FastCommuteObstruction}
$
\mathbb P\left[\sigma_d\leq \frac12 d^2\right]\leq \frac78.
$
\end{lemma}
We prove Lemma~\ref{lem:FastCommuteObstruction} later in Section~\ref{sec:ProofLemmaFastCommuteObstruction}. First, we show how this fast commute obstruction implies the desired quadratic lower bound.

\begin{lemma}
\label{lem:OneSidedHittingObstruction}
$
\mathbb P\left[\tau_d^+\leq \tfrac12d^2\right]\leq \frac{15}{16}.
$
\end{lemma}

\begin{proof}
Let $H_d(T)$ be the number of excursions of $R_t$ away from $0$ that reach level $d$
by time $T$. Conditional on a realization of the absolute trajectory
$(R_s)_{0\leq s\leq T}$, these excursions are fixed. If this realization contains exactly
$h$ such excursions, then $\{\tau_d^+\leq T\}$ occurs exactly when at least one of
these $h$ excursions has positive sign. Since departures from $0$ are symmetric and
independent, this conditional probability is
$
1-2^{-h}.
$
Therefore, whenever the realized trajectory $(R_s)_{0\leq s\leq T}$ contains exactly
$h$ excursions that reach level $d$, we have
$$
\mathbb E\left[
{\bf 1}_{\{\tau_d^+\leq T\}}
\mid
R_0,\ldots,R_T
\right]
=
1-2^{-h}.
$$
Therefore
\begin{align*}
\mathbb P\left[\tau_d^+\leq T\right]
=&
\sum_{h\geq 0}
\left(1-2^{-h}\right)\mathbb P[H_d(T)=h] \\
=& 
\frac12\mathbb P\left[H_d(T)=1\right]
+
\sum_{h\geq 2}
\left(1-2^{-h}\right)\mathbb P[H_d(T)=h] \\
\leq &
\frac12\mathbb P\left[H_d(T)=1\right]
+
\sum_{h\geq 2}
\mathbb P\left[H_d(T)=h\right] \\
= &
\frac12\mathbb P\left[H_d(T)=1\right]
+
\mathbb P\left[H_d(T)\geq 2\right] \\
\leq  &
\frac12\left( 1 - \mathbb P\left[H_d(T)\geq 2\right] \right)
+
\mathbb P\left[H_d(T)\geq 2\right] \\
= &
\frac12 + \frac12 \mathbb P\left[H_d(T)\geq 2\right].
\end{align*}

It remains to bound the second term. If $H_d(T)\geq 2$, then after the first visit of
the absolute chain to level $d$, the chain must hit $0$ and then hit $d$ again before
time $T$. Compare this motion with the reflected chain on $\{0,1,\ldots,d\}$, with the
same transition probabilities on $\{1,\ldots,d-1\}$, and with every attempted move
from $d$ to $d+1$ replaced by a move from $d$ to $d-1$. 
Starting from $d$, couple
the reflected chain and the original chain so that they make identical moves until
the original chain attempts to move from $d$ to $d+1$. At that moment, the reflected
chain moves from $d$ to $d-1$ instead. Thereafter, every excursion of the original
chain above $d$ is replaced in the reflected chain by the shorter path obtained by
immediately moving from $d$ to $d-1$. Consequently, the reflected chain reaches $0$
no later than the original chain. 
After returning to $0$, the evolution up to the next visit to $d$ is identical in
both chains.
Hence the reflected commute time
$\sigma_d$ is stochastically no larger than the corresponding time required by the
original absolute chain, that is
$
\mathbb P\left[H_d(T)\geq 2\right]
\leq
\mathbb P\left[\sigma_d\leq T\right].
$
Taking
$
T=\frac12d^2
$
and applying Lemma~\ref{lem:FastCommuteObstruction}, we obtain
$
\mathbb P\left[H_d(T)\geq 2\right]\leq \frac78
$, and the claim follows.
\end{proof}

The preceding lemma controls the probability that one walk reaches distance $d$ in a
specified direction by time $\frac12d^2$.
We now use this to obtain a positive lower bound on the probability that two independent walks together fail to cover distance $2d$ by that time.

\begin{lemma}
\label{lem:RangeObstruction}
Let $M_{1,t}^+$ and $M_{2,t}^+$ be two independent copies of
$
M_t^+:=\max_{0\leq s\leq t}X_s.
$
Then, 
$$
\mathbb P\left[M_{1,\frac12d^2}^+ + M_{2,\frac12d^2}^+ < 2d\right]
\geq
\frac1{256}.
$$
\end{lemma}

\begin{proof}
For $T=\frac12d^2$, and since $M_{1,T}^+$ and $M_{2,T}^+$ are independent, we have
$$
\mathbb P\left[M_{1,T}^+ + M_{2,T}^+ < 2d\right]
\geq
\mathbb P\left[M_{1,T}^+<d,\ M_{2,T}^+<d\right] 
=
\mathbb P\left[M_{1,T}^+<d\right]
\mathbb P\left[M_{2,T}^+<d\right].
$$
Since the walk moves in unit steps,
$
M_{i,T}^+\geq d
$
holds if and only if the walk hits level $d$ by time $T$, that is,
$
\{M_{i,T}^+\geq d\}=\{\tau_d^+\leq T\}
$. 
Therefore, for $i=1,2$, and using Lemma~\ref{lem:OneSidedHittingObstruction},
we have
$$
\mathbb P \left[M_{i,T}^+<d\right]
=
\mathbb P \left[\tau_d^+> T\right]
=
1- \mathbb P \left[\tau_d^+\leq T\right]
\geq 1/16.
$$
Combining the previous inequalities, we get
$\mathbb P\left[M_{1,T}^+ + M_{2,T}^+ < 2d\right] \geq (1/16)^2=1/256$, as wanted.
\end{proof}

We are now ready to prove Theorem~\ref{thm:QuadraticLowerBound}. 
\begin{proof}[Proof of Theorem~\ref{thm:QuadraticLowerBound}] 
Let
$
T=\frac12d^2
$,
and recall that
$
M_{X,T}^+=\max_{0\leq s\leq T}X_s,
$
and 
$
M_{Y,T}^-=\max_{0\leq s\leq T}(-Y_s)
$.
If rendezvous occurs by time $T$, then for some $s\leq T$,
we have 
$
X_s-Y_s=2d
$, and therefore
$$
2d=X_s-Y_s\leq M_{X,T}^+ + M_{Y,T}^-.
$$
It follows that 
$
\{\mathcal T_d > T\}
\supseteq
\{M_{X,T}^+ + M_{Y,T}^- < 2d\}.
$
The one-agent signed displacement process is invariant under reflection: if $(X_t)_{t\geq0}$ follows the process, then $(-X_t)_{t\geq0}$ has the same distribution. Indeed, the transition probability to move right from state $i>0$ equals $a_i$, while from state $-i$ it equals $1-a_i$. Consequently,
$
M_{Y,T}^-=\max_{0\leq s\leq T}(-Y_s)
$
has the same distribution as an independent copy of
$
M_{X,T}^+.
$
Applying Lemma~\ref{lem:RangeObstruction}, we get
$$
\E[\mathcal T_d]
\geq
T~\mathbb P [\mathcal T_d>T]
\geq
T~\mathbb P [M_{X,T}^+ + M_{Y,T}^- < 2d]
\geq
\frac{T}{256}.
$$
Since $T=\tfrac12 d^2$, we just showed that 
$R_d = \inf_{a \in \mathcal A} \E[\mathcal T_d] \geq \tfrac{1}{512}d^2$, as wanted. 
\end{proof}

\subsubsection{The Proof of Lemma~\ref{lem:FastCommuteObstruction}}
\label{sec:ProofLemmaFastCommuteObstruction}
In this section we prove Lemma~\ref{lem:FastCommuteObstruction}. 
If $a_i=0$ or $a_i=1$ for some $1\leq i\leq d-1$, 
then the reflected chain cannot complete the commute $d\to0\to d$, and that would imply $\sigma_d=\infty$, hence the conclusion would be trivial. Therefore, we may assume that $0<a_i<1$, for all $i=1,\ldots,d-1$. 

The proof of Lemma~\ref{lem:FastCommuteObstruction} further reduces to the following claim.
\begin{lemma}
\label{lem:BoundsEsigmaEsigma2}
The following bounds hold:
\begin{enumerate}[(a)]
\item
\label{item:EsigmaLowerBound}
$
\E[\sigma_d]
\geq
d^2.
$
\item
\label{item:Esigma2Bound}
$
\E[\sigma_d^2]
\leq
2\left(\E[\sigma_d]\right)^2.
$
\end{enumerate}
\end{lemma}

We prove Lemma~\ref{lem:BoundsEsigmaEsigma2} below, but first we demonstrate how that implies Lemma~\ref{lem:FastCommuteObstruction}. 

\begin{proof}[Proof of Lemma~\ref{lem:FastCommuteObstruction}]
Our goal is to show that $
\mathbb P\left[\sigma_d\leq \frac12 d^2\right]\leq \frac78.
$
Using the bound
$
\E[\sigma_d]\geq d^2
$
from Lemma~\ref{lem:BoundsEsigmaEsigma2}, we have
$$
\mathbb P\left[\sigma_d\leq \tfrac12 d^2\right]
\leq
\mathbb P\left[\sigma_d\leq \tfrac12 \E[ \sigma_d] \right]
= 
1-
\mathbb P\left[\sigma_d> \tfrac12 \E[ \sigma_d] \right]
$$

To further upper bound the last expression, we apply the Paley-Zygmund inequality to the nonnegative random variable $\sigma_d$, and use also $\E[\sigma_d^2]\leq 2(\E[\sigma_d])^2$ from Lemma~\ref{lem:BoundsEsigmaEsigma2}, to get
$$
\mathbb P\left[\sigma_d>\theta\E[\sigma_d]\right]
\geq
\frac{(1-\theta)^2(\E[\sigma_d])^2}{\E[\sigma_d^2]}
\geq
\frac{(1-\theta)^2}{2}.
$$
The desired bound is now obtained by choosing $\theta=1/2$.
\end{proof}

Therefore, it remains to prove Lemma~\ref{lem:BoundsEsigmaEsigma2}.
Part~\ref{item:EsigmaLowerBound} is an easy application of Lemma~\ref{lem:ElectricalCommuteIdentity}. 

\begin{proof}[Proof of Lemma~\ref{lem:BoundsEsigmaEsigma2}~\ref{item:EsigmaLowerBound}.]

Recall that $\sigma_d$ is the commute time from $d$ to $0$ and back to $d$ in the truncated absolute-value chain $(\widetilde R_t)$. Hence
$
\E[\sigma_d]
=
\E[T_{d\to0}]
+
\E[T_{0\to d}].
$

Consider the chain $(\widetilde R_t)_{t\geq0}$ on $\{0,1,\ldots,d\}$ obtained from $(R_t)$ by replacing the transition from $d$ to $d+1$ with a deterministic move from $d$ to $d-1$, as defined before Lemma~\ref{lem:FastCommuteObstruction}. Its transition probabilities are
$
P(0,1)=1
$
and
$
P(d,d-1)=1,
$
while, for $1\leq i\leq d-1$,
$
P(i,i+1)=a_i >0
$
and
$
P(i,i-1)=1-a_i >0
$
Therefore, the birth-death chain is reversible. 
Let $c_i$ denote the induced conductance of the edge $(i,i+1)$, so that the total conductance is
$
c_G
=
\sum_{i=0}^{d-1}c_i,
$
and the 
effective resistance
$
R(0\leftrightarrow d)
=
\sum_{i=0}^{d-1}
\frac{1}{c_i}.
$
Lemma~\ref{lem:ElectricalCommuteIdentity} then gives 
$$
\E[\sigma_d]
=
\E[T_{d\to0}]
+
\E[T_{0\to d}]
=
2
c_G
R(0\leftrightarrow d)
\geq 
\left(\sum_{i=0}^{d-1}c_i\right)
\left(\sum_{i=0}^{d-1}\frac1{c_i}\right)
\geq
\left(\sum_{i=0}^{d-1}\sqrt{c_i}\frac1{\sqrt{c_i}}\right)^2
=
d^2,
$$
where the inequality follows from the Cauchy--Schwarz inequality.
\end{proof}

The remaining of the section builds up towards the proof of Lemma~\ref{lem:BoundsEsigmaEsigma2}~\ref{item:Esigma2Bound}. 
For $1\leq i\leq d$, let $D_i$ denote the hitting time of $i-1$ for the reflected chain $(\widetilde R_t)$ started from $i$.
Starting from $d$, the chain must successively move from $d$ to $d-1$, then from $d-1$ to $d-2$, and so on until it reaches $0$. 
Therefore, by the strong Markov property,
$$
T_{d\to0}=D_d+D_{d-1}+\cdots+D_1,
$$
where the random variables $D_1,\ldots,D_d$ are independent.
Define
$$
m_i=\E[D_i],
~~~\text{and}~~
v_i=\operatorname{Var}(D_i).
$$
Since the chain moves deterministically from $d$ to $d-1$, we have
$D_d=1, m_d=1$, and $v_d=0$.

Now fix $1\leq i\leq d-1$. Starting from $i$, the chain either moves directly to $i-1$, which completes $D_i$ in one step, or it first moves to $i+1$. In the latter case, it must return from $i+1$ to $i$ before it can try again to move from $i$ to $i-1$. Thus $D_i$ is obtained by repeating independent attempts until the first successful move from $i$ to $i-1$.
Let 
$
A_i=\frac{a_i}{1-a_i}, 
B_i=\frac{1}{1-a_i}
$.
We have the following recursive statements. 
\begin{lemma}
\label{lem:recursionOnMiVi}
For $1\leq i\leq d-1$, the following recursive identities hold:
\\
(a) $
m_i
=
1+A_i\left(1+m_{i+1}\right)
=
B_i+A_i m_{i+1}.
$ \\
(b) 
$
v_i
=
A_i v_{i+1}
+
A_iB_i\left(1+m_{i+1}\right)^2$.
\end{lemma}

\begin{proof}
Before treating each of (a), (b), individually, we need some preparatory observations. 
Let $G_i$ be the number of failed attempts before the chain first succeeds in moving from $i$ to $i-1$. 
Since $G_i$ counts the number of failed attempts before the first successful move from $i$ to $i-1$, it is a geometric random variable with parameter $1-a_i$. Thus
$
\mathbb P\left[G_i=k\right]=a_i^k(1-a_i),
$
for $k=0,1,2,\ldots$. Therefore
$
\E[G_i]
=
\frac{a_i}{1-a_i}
=
A_i,
$
and
$
\operatorname{Var}(G_i)
=
\frac{a_i}{(1-a_i)^2}
=
A_iB_i.
$

For each failed attempt $\ell$, let $D_{i+1}^{(\ell)}$ be the time needed for the chain to return from $i+1$ to $i$ after that failed attempt. The random variables
$
D_{i+1}^{(1)},D_{i+1}^{(2)},\ldots
$
are independent copies of $D_{i+1}$, and they are independent of $G_i$.
We claim that
\begin{equation}
\label{equa:RecOnDi}
D_i
=
1+\sum_{\ell=1}^{G_i}\left(1+D_{i+1}^{(\ell)}\right).
\end{equation}
Indeed, the final $1$ is the successful step from $i$ to $i-1$, while each failed attempt contributes one step from $i$ to $i+1$ and then a return time from $i+1$ to $i$, hence the claim follows. 

We are now ready to prove part (a) of the claim. Taking expectations in~\eqref{equa:RecOnDi} we get 
$
m_i
=
1+\E[G_i]\left(1+m_{i+1}\right)
$.
Since $\E[G_i]=A_i$, and $B_i=A_i+1$, claim (a) follows. 

For claim (b), we take the variance on both hand-sides of~\eqref{equa:RecOnDi}. 
Since the random variables
$1+D_{i+1}^{(1)},
1+D_{i+1}^{(2)},
\ldots
$
are independent and identically distributed, and are independent of $G_i$, the variance formula for a random sum gives
\begin{align*}
v_i 
&=
\operatorname{Var}\left(1+\sum_{\ell=1}^{G_i}\left(1+D_{i+1}^{(\ell)}\right)\right)\\
&=
\operatorname{Var}\left(\sum_{\ell=1}^{G_i}\left(1+D_{i+1}^{(\ell)}\right)\right)\\
&=
\E[G_i]\operatorname{Var}\left(1+D_{i+1}\right)
+
\operatorname{Var}(G_i)\left(\E\left[1+D_{i+1}\right]\right)^2\\
&=
\E[G_i]\operatorname{Var}\left(D_{i+1}\right)
+
\operatorname{Var}(G_i)\left(1+\E[D_{i+1}]\right)^2\\
&=
\E[G_i]v_{i+1}
+
\operatorname{Var}(G_i)\left(1+m_{i+1}\right)^2\\
&=
A_i v_{i+1}
+
A_iB_i\left(1+m_{i+1}\right)^2. 
\end{align*}

\end{proof}

We now use the above recurrences to establish an upper bound on the variance of the time needed to go from $i$ to $i-1$. For this, we let $S_i=\sum_{j=i}^d m_j$.
\begin{lemma}
\label{lem:ViUpperBoundLocal}
For $i=1,\ldots, d$, we have 
$
v_i\leq m_i^2+2m_iS_{i+1}
$, where $S_{d+1}=0$. 
\end{lemma}

\begin{proof}
We prove the claim by backward induction on $i$. 
For $i=d$, this says $0\leq 1$, which is true.

For the inductive step, assume the inequality holds for $i+1$. 
For notational simplicity, set 
$
A=A_i,
B=B_i,
m=m_{i+1}
$,
and 
$
R=S_{i+2}
$. 
Then we have 
$
S_{i+1}=m+R,
$
and the induction hypothesis gives
$
v_{i+1}\leq m^2+2mR
$.
Since $A\geq 0$, multiplying both hand sides of the induction hypothesis
by $A$ , and then adding the same non-negative term
$
AB(1+m)^2
$
gives
$$
A v_{i+1}
+
AB(1+m)^2
\leq
A\left(m^2+2mR\right)
+
AB(1+m)^2.
$$
By Lemma~\ref{lem:recursionOnMiVi}(b), the left-hand side of the last inequality equals $v_i$. 
Hence, it remains to further upper bound the right hand side of the last inequality by $m_i^2+2m_iS_{i+1}$. To that end, using Lemma~\ref{lem:recursionOnMiVi}(a), we have 

\begin{align*}
&m_i^2+2m_iS_{i+1}
-
\left(
A\left(m^2+2mR\right)
+
AB(1+m)^2
\right)\\
&=
(B+Am)^2+2(B+Am)(m+R)
-
\left(
A\left(m^2+2mR\right)
+
AB(1+m)^2
\right)\\
&=
B^2+2Bm+2BR-AB\\
&=
B(B-A+2m+2R)\\
&=
B(1+2m+2R),
\end{align*}
which is non-negative. Here we used $B=A+1$ both to cancel the $m^2$ terms and, in the last equality, to obtain $B-A=1$.
\end{proof}

We are ready to conclude the remaining technical claim.  
\begin{proof}[Proof of Lemma~\ref{lem:BoundsEsigmaEsigma2}~\ref{item:Esigma2Bound}]
Summing the inequalities of Lemma~\ref{lem:ViUpperBoundLocal}, from $i=1$ to $d$, we obtain
$$
\operatorname{Var}(T_{d\to 0})
=
\sum_{i=1}^d v_i
\leq
\sum_{i=1}^d \left(m_i^2+2m_iS_{i+1}\right)
=
\left(\sum_{i=1}^d m_i\right)^2
=
\left(\E [T_{d\to 0}]\right)^2.
$$
The same argument applied to $T_{0\to d}$ gives
$
\operatorname{Var}(T_{0\to d})
\leq
\left(\E [T_{0\to d}]\right)^2
$.
Since now $T_{d\to 0}$ and $T_{0\to d}$ are independent,
$$
\operatorname{Var}(\sigma_d)
=
\operatorname{Var}(T_{d\to 0})
+
\operatorname{Var}(T_{0\to d})
\leq
\left(\E [T_{d\to 0}]\right)^2+
\left(\E [T_{0\to d}]\right)^2
\leq 
\left(\E [T_{d\to 0}]+ \E [T_{0\to d}] \right)^2
=
\left(\E[\sigma_d]\right)^2.
$$
But then, 
$$
\E[\sigma_d^2]
=
\operatorname{Var}(\sigma_d)
+
\left(\E[\sigma_d]\right)^2
\leq 
2\left(\E[\sigma_d]\right)^2
$$
as wanted.

\end{proof}

\subsection{Nearly Quadratic Upper Bound for Oblivious Strategies (Unknown-Distance Setting)}
\label{sec:QuadraticUpperBound}

We next prove that oblivious self-distance strategies can achieve nearly quadratic expected rendezvous time. The strategy used in the proof is independent of $d$, and therefore it is applicable to the unknown-distance setting. 

\begin{theorem}
\label{thm:NearQuadraticUpperBound}
For every $\eta>0$,
there is an admissible strategy $a\in \mathcal A$ such that for all $d\geq 1$,
$
\E[\mathcal T_d(a)]=O\left(d^{2+\eta}\right).
$
\end{theorem}

It is immediate that no bounded-support admissible strategy results in feasible rendezvous for big enough $d$. Hence, the proof of the theorem naturally relies on the following infinite-support \emph{asymptotically balanced inward-drift} admissible strategy, which is independent of the value of $d$, but is tailored to some fixed $\epsilon\in(0,1/4)$.
\begin{definition}
\label{def:BalancedInward}
For $\epsilon \in \left(0,\tfrac14\right)$, let $c=\tfrac14+\epsilon$. We set
$a_0=\frac12$, and $a_i=\frac12-\frac{c}{i}$, for $i\geq1$. 
\end{definition}
For the fixed $\epsilon>0$, we show that for some constant $\gamma=\gamma(\epsilon)$, the strategy $a=a(\epsilon)$ of Definition~\ref{def:BalancedInward} satisfies
$
\E[\mathcal T_d]
\leq
\gamma~d^{2+4\epsilon}.
$
This implies Theorem~\ref{thm:NearQuadraticUpperBound} by choosing $\epsilon$ so that $4\epsilon<\eta$.

An important feature of this construction is that the strategy is defined independently of $d$. Thus the same upper bound also applies in the more general oblivious self-distance model, where rendezvous is attempted without prior knowledge of the initial distance $2d$. A more careful tracking of the estimates also shows that the hidden constant grows at the order of
$
O\left(\epsilon^{-2}\right).
$
Since this dependence is not important for our purposes, we do not optimize it here.

The strategy of Definition~\ref{def:BalancedInward} introduces a slowly decaying inward drift. When an agent is at distance $i$ from its starting point, the probability of moving further away is reduced by an amount of order $1/i$. 
This drift is weak enough to preserve exploration, yet strong enough to force repeated returns to the origin while still allowing polynomially small mass at large distances.

The proof proceeds in two steps. We first analyze the one-agent displacement chain and show that its mass decays polynomially with the distance from the origin. This decay determines the conductances of the associated two-agent electrical network, and therefore quantifies how costly it is to move through distant configurations.
We then use this information in the two-agent chain. The key point is that we do not need to analyze all possible ways in which rendezvous may occur. In the electrical-network formulation, it is enough to upper bound the effective resistance from the initial state to the meeting set. This can be done by restricting attention to one explicit route leading to rendezvous, since removing all other possible routes can only increase the resistance. The polynomial decay obtained from the one-agent analysis is then strong enough to show that the resistance of this route grows only polynomially in $d$, which yields the desired upper bound.

We start by studying the one-agent displacement chain induced by the admissible strategy of Definition~\ref{def:BalancedInward}.
By Section~\ref{sec:ElectricalRepresentationInducedChains}, this chain is reversible. Let $\pi$ denote a reversible measure. Then
$
\pi(i)a_i=\pi(i+1)(1-a_{i+1})
$
for every $i\geq 0$.
The following lemma controls the growth of $\pi(i)$. 

\begin{lemma}
\label{lem:PolynomialTailInvariantMeasure}
For every fixed $\epsilon \in (0,1/4)$, and for all $i\in\mathbb Z\setminus\{0\}$, we have
$
\pi(i) = \Theta\left( 
|i|^{-(1+4\epsilon)}
\right)
$.
\end{lemma}

\begin{proof}
By symmetry, it is enough to define $\pi$ on $\integers_{\geq0}$ and then set $\pi(-i)=\pi(i)$. For arbitrarily fixed $\pi(0)>0$, set
$
(1-a_1)\pi(1)=a_0\pi(0),
$
and, for $j\geq1$, define
$
\frac{\pi(j+1)}{\pi(j)}
=
\frac{a_j}{1-a_{j+1}}
=
\frac{(j-2c)(j+1)}{j(j+1+2c)}.
$
Then, for $j\geq1$, we have
$$
\frac{\pi(j)}{\pi(0)}
=
\frac{a_0}{1-a_1}
\prod_{k=1}^{j-1}\frac{k-2c}{k}
\prod_{k=1}^{j-1}\frac{k+1}{k+1+2c}
=
\frac{\Gamma(1+2c)}{\Gamma(1-2c)}
~\frac{\Gamma(j-2c)\Gamma(j+1)}
{\Gamma(j)\Gamma(j+1+2c)},
$$
where the equalities used that $a_0=1/2$, $1-a_1=1/2+c$, and the Gamma product identity (note that all terms are well defined, because $0<c<1/2$).
By Stirling's approximation, we obtain
$$
\frac{\Gamma(j-2c)}{\Gamma(j)}
=
\Theta\left(j^{-2c}\right),
\quad
\frac{\Gamma(j+1)}{\Gamma(j+1+2c)}
=
\Theta\left(j^{-2c}\right),
$$
and therefore 
$
\pi(j)
=
\Theta\left(j^{-4c}\right)
=
\Theta\left(j^{-(1+4\epsilon)}\right).
$
\end{proof}

We are now ready to conclude with the proof of Theorem~\ref{thm:NearQuadraticUpperBound}.

\begin{proof}[Proof of Theorem~\ref{thm:NearQuadraticUpperBound}]
Fix $\eta>0$, and choose $\epsilon\in(0,1/4)$ such that $4\epsilon<\eta$.
For the strategy of Definition~\ref{def:BalancedInward}, we consider the path
$$
\mathcal P: ~~S_{0,0}
\to
S_{1,-1}
\to
S_{2,0}
\to
S_{3,-1}
\to
\cdots
\to
S_{2d-1,-1}.
$$
The final state belongs to $\mathcal A_d$, since
$
(2d-1)-(-1)=2d.
$
Let
$
\Pi(x,y)=\pi(x)\pi(y),
$
and consider the two-agent chain in which the two coordinates continue to move independently according to the strategy even after a meeting state is reached. Since the one-agent chain is reversible with respect to $\pi$, this two-agent chain is reversible with respect to $\Pi$. Moreover, its first hitting time of $\mathcal A_d$ is exactly $\mathcal T_d$.
Applying Lemma~\ref{lem:InfiniteResistanceHittingBound} with starting state $S_{0,0}$ and target set
$
A=\mathcal A_d
$
gives
$$
\E[\mathcal T_d]
\leq
\E[T_{S_{0,0} \to A}]
\leq
2c_G R(S_{0,0}\leftrightarrow\mathcal A_d).
$$
Next we show that $c_G=O(1)$ and $R(S_{0,0}\leftrightarrow\mathcal A_d)=O\left(d^{2+4\epsilon}\right)$ concluding the theorem. 
For this, we first bound the total conductance of the network, and then bound the conductances of the edges along the chosen path.

We start with the total conductance. 
Since every transition probability is at most $1$, we get
$$
2c_G
\leq
\sum_{S_{x,y},S_{u,v}}
\Pi(x,y)P(S_{x,y},S_{u,v})
\leq
\sum_{S_{x,y}}
\Pi(x,y)
=
\sum_{x,y\in\mathbb Z}\pi(x)\pi(y)
=
\left(\sum_{k\in\mathbb Z}\pi(k)\right)^2.
$$
By Lemma~\ref{lem:PolynomialTailInvariantMeasure},
$
\pi(k)
=
\Theta\left(|k|^{-(1+4\epsilon)}\right)
$,
and therefore
$
\sum_{k\in\mathbb Z}\pi(k)<\infty
$, implying that 
$
c_G=O(1)
$.

Next we compute the total effective resistance. 
Removing all edges not belonging to this path can only increase the effective resistance, that is 
$
R(S_{0,0}\leftrightarrow\mathcal A_d)
\leq
\sum_{e \in \mathcal P}
\tfrac1{c_e},
$
where the sum is over the edges of the above path.
Since $a_0=1/2$, the first edge has conductance 
$
c_e=\Pi(0,0)/4
=
\pi(0)^2/4,
$
which is bounded below by a positive constant depending only on $\epsilon$.
Now consider any other edge. If the larger coordinate is indexed by $k$, then the other coordinate belongs to $\{0,-1\}$. 
Since $\pi(0)$ and $\pi(-1)$ are fixed positive constants, and by Lemma~\ref{lem:PolynomialTailInvariantMeasure},
$
\pi(k)=\Theta(k^{-(1+4\epsilon)}),
$
we obtain that 
$
\Pi(x,y)
=
\Theta\left(k^{-(1+4\epsilon)}\right).
$
Moreover, apart from the first edge already considered, every transition probability along the path is either $b_k(1-b_0)$ or $b_kb_{-1}$, for some $k\geq1$. Since
$
b_k\geq \frac14-\epsilon,
1-b_0=\frac12,
b_{-1}=\frac34+\epsilon,
$
these probabilities are bounded below by a positive constant depending only on $\epsilon$, and above by $1$. Therefore
$
c_e
=
\Pi(x,y)P(S_{x,y},S_{u,v})=
\Theta\left(k^{-(1+4\epsilon)}\right)
$.
Then, we have 
$$
R(S_{0,0}\leftrightarrow\mathcal A_d)
\leq
\sum_{e \in \mathcal P} \frac1{c_e}
=
O\left(1+\sum_{k=1}^{2d-1} k^{1+4\epsilon}\right)
=
O\left(d^{2+4\epsilon}\right),
$$
as wanted. 
\end{proof}

The significance of the chosen constants of the admissible strategy of Definition~\ref{def:BalancedInward} is now self-evident. 
More generally, if one takes
$
a_i=\frac12-\frac{c}{i},
$
then the corresponding reversible measure satisfies
$
\pi(i)=\Theta(i^{-4c}).
$
Hence the total conductance is finite only when $4c>1$, that is,
$c>1/4$. In that case, our arguments show an upper bound of order
$
d^{1+4c}.
$
Thus $c=1/4$ is the critical threshold, and choosing
$
c=\frac14+\epsilon
$
gives the smallest exponent obtainable from this method, namely
$
d^{2+4\epsilon}.
$

\subsection{Quadratic Upper Bound for Oblivious Strategies (Known-Distance Setting)}
\label{sec:AbsQuadraticUpperBound}

In this section we prove a pure quadratic upper bound in the known-distance setting. 
For this, we use the \emph{level-$n$ uniform truncated strategy}, defined by
$
a_i=\frac12
$
for $0\leq i<n$, 
and $a_i=0$ for $i\geq n$.
This is an admissible strategy, and it induces the level-$n$ truncated chain of
Definition~\ref{def:n-TruncatedChain}.  By Lemma~\ref{lem:UpperBoundFrame}, we have that $R_d \leq U_{d,n}(a)$, and therefore, our pure quadratic bound is immediate by the following theorem.

\begin{theorem}
\label{thm:PureQuadraticTruncatedUniform}
For $d\geq 1$, let $a$ be the level-$2d$ uniform truncated strategy. Then
$
U_{d,2d}(a)=O(d^2).
$
\end{theorem}

We use the following standard estimate for the one-dimensional coordinate walk
induced by the level-$n$ uniform truncated strategy.

\begin{lemma}
\label{lem:QuadraticHittingTime}
Consider the Markov chain on $\{-n,\ldots,n\}$ with transition probabilities
$
P(i,i+1)=P(i,i-1)=\frac12
$
for $-n<i<n$,
and
$
P(n,n-1)=1, P(-n,-n+1)=1.
$
The expected hitting time between any two states of this chain is $O(n^2)$.
\end{lemma}
\begin{proof}
The chain is reversible with respect to $\pi(-n)=\pi(n)=1$ and $\pi(i)=2$ for $|i|<n$. Indeed, for every edge $(i,i+1)$, we have
$
c(i,i+1)=\pi(i)P(i,i+1)=1,
$
since this equals $1\cdot 1$ on the edge $(-n,-n+1)$, equals $2\cdot \frac12$ on all interior edges, and equals $2\cdot \frac12$ on the edge $(n-1,n)$.
That is, every edge conductance satisfies
$
c_e=1
$.

Let $u,v\in\{-n,\ldots,n\}$ be arbitrary states.
Applying Lemma~\ref{lem:ElectricalCommuteIdentity}, we obtain
$$
\E[T_{u\to v}]
+
\E[T_{v\to u}]
=
2c_G R(u\leftrightarrow v).
$$
Next we show that
$
c_G=O(n)
$
and
$
R(u\leftrightarrow v)=O(n).
$
Since for every edge $e$ we have $c_e=1$, and the network contains $2n$ edges, we get
$
c_G
=
\sum_e c_e
=
2n.
$

Similarly, every edge resistance equals $1/c_e=1$. Since the distance between any two
states is at most $2n$,
$$
R(u\leftrightarrow v)
=
\sum_{e\in Path(u,v)}
\frac1{c_e}
\leq
2n,
$$
where $Path(u,v)$ denotes the unique path between $u$ and $v$.
It follows that both $\E[T_{u\to v}], \E[T_{v\to u}]$ are finite, and 
$$
\E[T_{u\to v}]
=
2c_G R(u\leftrightarrow v)
-\E[T_{v\to u}]
\leq 
2c_G R(u\leftrightarrow v)
=
O(n^2).
$$
\end{proof}

The proof of Theorem~\ref{thm:PureQuadraticTruncatedUniform} is based on a repeated-trial argument for the level-$2d$ uniform truncated strategy. In each trial, we wait until the left agent reaches physical location $d$, equivalently until her displacement coordinate reaches $2d$. By Lemma~\ref{lem:QuadraticHittingTime}, this requires expected time $O(d^2)$. Since the right agent starts the trial at her own origin, symmetry implies that, at this stopping time, she is at or to the left of her origin with probability at least $1/2$. On this event, rendezvous must already have occurred. Otherwise, we wait until the right agent returns to her origin, which again requires expected time $O(d^2)$, and begin the next trial. Thus each trial has expected duration $O(d^2)$ and succeeds with probability at least $1/2$. Summing over the resulting geometric sequence yields $U_{d,2d}(a)=O(d^2)$. We now formalize this argument.

\begin{proof}[Proof of Theorem~\ref{thm:PureQuadraticTruncatedUniform}]
Set $n=2d$.  For the level-$n$ uniform truncated strategy, each coordinate $X_t,Y_t$ of the
chain of Definition~\ref{def:n-TruncatedChain} evolves as the one-dimensional
chain of Lemma~\ref{lem:QuadraticHittingTime}, until the absorbing condition
$x-y=2d$ is reached.

We use repeated trials.  At the beginning of each trial, the second coordinate (right agent) is initially at $0$ (physical coordinate $d$). The next argument holds true for any initial displacement of the first coordinate.
Starting from the current value of the first coordinate (left agent) starting at any physical location $j<d$ (initially $j=-d$), we wait until the first
coordinate hits $n=2d$.  By Lemma~\ref{lem:QuadraticHittingTime}, this takes expected
time $O(d^2)$, for any starting value $j$. 

Let $H$ denote this hitting time.  The time $H$ depends only on the first coordinate of the chain.
Since the second coordinate starts the trial at $0$, and since the one-dimensional
coordinate walk is symmetric around $0$, we have
$
\mathbb P[Y_H\leq 0]\geq \frac12.
$
On this event, rendezvous has already occurred.  Indeed, at time $H$ the first
agent is at physical position
$
-d+n=d,
$
while the second agent is at physical position
$
d+Y_H\leq d.
$
Equivalently, in the chain notation, $X_H-Y_H\geq 2d$, and since the difference
changes by steps of size $-2,0,$ or $2$, the absorbing state must have been hit
by time $H$.

If $Y_H>0$ (which happens with probability at most $1-\mathbb P[Y_H\leq 0]\leq 1/2$), we wait until the second coordinate returns to $0$. 
Again, by Lemma~\ref{lem:QuadraticHittingTime}, this extra waiting time has expectation $O(d^2)$. 
At this point the second coordinate is again $0$, hence the same trial argument applies again by the strong Markov property.
Therefore each trial has expected length $O(d^2)$, and each trial succeeds with probability at least $1/2$. Repeating the argument after each failed trial, and applying the strong Markov property, we obtain
$$
U_{d,2d}(a)
\leq
\sum_{k\geq0} O(d^2)\,2^{-k}
=
O(d^2).
$$
\end{proof}

\subsubsection{Numerical Estimates of the Constants in the Quadratic Upper Bound}
\label{sec:NumEstimatesQuadraticConstants}

Theorem~\ref{thm:PureQuadraticTruncatedUniform} established that
$
R_d\leq U_{d,2d}(a)=O(d^2),
$
where $a$ is the level-$2d$ uniform truncated strategy.
The proof was deliberately qualitative, and did not attempt to optimize the hidden
constant in the quadratic bound.
On the other hand, together with the lower bound of Theorem~\ref{thm:QuadraticLowerBound} we obtain that $R_d = \Theta(d^2)$. 
The numerical investigations of
Sections~\ref{sec:UpperBoundsSmallD1} and~\ref{sec:NumericalComputations}
suggest that the corresponding constant of the expected rendezvous time is relatively small.
Our goal in this section is therefore not to estimate the optimal rendezvous value
$R_d$, but rather to investigate bounds to the asymptotic constant obtainable within the much simpler family of uniform truncated strategies.

More precisely, for each fixed $d\geq1$, we restrict attention in this numerical study to truncation levels $n>d$, and consider the level-$n$ uniform truncated strategy. The boundary case $n=d$, although allowed by Definition~\ref{def:n-TruncatedChain}, does not affect the asymptotic question considered here and changes only the $d=1$ entry among the values reported below.
For this strategy, the expected rendezvous time is exactly
$U_{d,n}(a)$, as defined in Section~\ref{sec:TruncatedChains}.
Theorem~\ref{thm:PureQuadraticTruncatedUniform} corresponds to the particular
choice $n=2d$.
The natural question then is whether a different truncation level produces a smaller
quadratic constant. More specifically, if 
$$
n^\star(d)
=
\arg\min_{n>d} U_{d,n}(a),
$$
where $n^\star=n^\star(d)$, 
does the ratio
$
n^\star/d,
$
approach a constant as $d\to\infty$?
If so, what is the limiting value
$
\lambda^\star
=
\lim_{d\to\infty}\frac{n^\star}{d},
$
and what quadratic constant is obtained by the corresponding strategy?

To investigate these questions, we computed
$
U_{d,n}(a)
$
for all
$
1\leq d\leq 500
$. As an additional large-scale verification, we also computed the value corresponding to $d=1000$. 

Unlike the weak-peek computations of
Section~\ref{sec:NumericalComputations}, no nonlinear optimization is required.
For every fixed pair $(d,n)$, the strategy is completely determined, and the value
$U_{d,n}(a)$ is obtained by solving the finite system of hitting-time equations
from Lemma~\ref{lem:finite-hitting-system} for the level-$n$ truncated chain of
Definition~\ref{def:n-TruncatedChain}.
The computations were implemented in Julia.
Rather than constructing the full sparse matrix of the linear system,
we exploited the block-tridiagonal structure of the resulting linear system.
This allows the system defining the values
$E_{x,y}$
to be solved by repeated block elimination, reducing both memory usage and running
time substantially.
For each fixed pair $(d,n)$, the resulting value
$
U_{d,n}(a)=E_{0,0}
$
was computed up to numerical linear-algebra precision.
For every distance $d$, we then searched over a range of truncation levels
$n>d$, recording both the minimizing value
$n^\star$
and the corresponding normalized expected rendezvous time
$
U_{d,n^\star}(a)/d^2.
$

The computations reveal two phenomena. 
First, the optimal truncation ratio
$
n^\star/d
$
stabilizes rapidly as $d$ grows.
Second, the normalized expected rendezvous time
$
U_{d,n^\star}(a)/d^2
$
also stabilizes rapidly.
Table~\ref{tab:uniform-truncation-summary} reports representative values of the
optimal truncation level and the corresponding normalized expected rendezvous time.
The complete dataset for $1\leq d\leq 500$ is omitted for brevity.  In addition,
we include several larger values of $d$ in order to illustrate the stability of the
observed asymptotic regime.

\begin{table}[ht]
\centering
\small
\begin{tabular}{c|cccccccccc}
\hline
$d$
&
1
&
2
&
5
&
10
&
20
&
50
&
100
&
200
&
500
&
1000
\\
\hline
$n^\star$
&
2
&
3
&
9
&
18
&
36
&
90
&
181
&
361
&
903
&
1806
\\
\hline
$n^\star/d$
&
2.000
&
1.500
&
1.800
&
1.800
&
1.800
&
1.800
&
1.810
&
1.805
&
1.806
&
1.806
\\
\hline
$U_{d,n^\star}(a)/d^2$
&
7.6316
&
8.0771
&
8.5373
&
8.6446
&
8.6872
&
8.7069
&
8.7118
&
8.7137
&
8.7146
&
8.7149
\\
\hline
\end{tabular}
\caption{
Representative values of the optimal truncation level $n^\star=n^\star(d)$,
the optimal truncation ratio $n^\star/d$, and the corresponding
normalized expected rendezvous time
$U_{d,n^\star}(a)/d^2$
for the level-$n$ uniform truncated strategy.
}
\label{tab:uniform-truncation-summary}
\end{table}

The data suggest that both
$
n^\star/d
$
and
$
U_{d,n^\star}(a)/d^2
$
stabilize rapidly.
In particular, the optimal truncation ratio appears to converge to a value close to
$
\lambda^\star \approx 1.806,
$
while the corresponding normalized expected rendezvous time appears to converge to
$
C^\star \approx 8.715.
$
Although these computations do not determine the optimal rendezvous value $R_d$,
they provide numerical evidence that the level-$n$ uniform truncated strategy admits a well-defined asymptotic quadratic constant.

\section{Conclusion \& Future Directions}
\label{sec:Conclusion}
We introduced oblivious self-distance strategies for rendezvous on the line and studied both the known-distance and unknown-distance settings. For known distance, we proved that the optimal expected rendezvous time satisfies
\[
R_d=\Theta(d^2).
\]
The lower bound also applies to the unknown-distance setting. For the latter, we established a universal strategy, independent of $d$, with expected rendezvous time
\[
O(d^{2+\eta})
\]
for every fixed $\eta>0$.

For $d=1$, our truncated and weak-peek models determine the exact optimum. For $d=2,\ldots,6$, these models give rigorous finite-support upper bounds, while extensive numerical optimization repeatedly identifies the same truncation structure and objective values. We conjecture that the corresponding strategies are optimal, although global optimality of the weak-peek minimizers is not proved. The finite-state analysis is based on the weak-peek relaxation, while our asymptotic results are obtained separately using probabilistic arguments based on birth-death chains and electrical-network methods.

Several natural questions remain open. The most immediate concerns the unknown-distance setting: does there exist a universal oblivious strategy, independent of $d$, with expected rendezvous time
$
\Theta(d^2)?
$
Even in the known-distance setting, the optimal quadratic constant is unknown. Our numerical results suggest that the true asymptotic constant is substantially smaller than the bounds obtained by our proofs. We conjecture that the finite-support strategies reported for $d=2,\ldots,6$ are optimal and, more generally, that every $d$ admits an optimal bounded-support strategy. Certifying global optimality of the corresponding weak-peek minimizers remains open, as does extending the weak-peek methodology to obtain exact solutions for larger values of $d$.

More broadly, the weak-peek framework provides a systematic way to derive finite-state relaxations of infinite-state rendezvous problems. It would be interesting to understand whether similar relaxations can be applied to other rendezvous and search models.


\bibliography{ObliviousRend-arXiv-clean2}

\appendix

\section{Computer-Assisted Verification for Lemma~\ref{lem:d1level3Minimizer}}
\label{sec:MathematicaCode}

For completeness, we give the exact Mathematica input used to exclude
interior critical points in the proof of Lemma~\ref{lem:d1level3Minimizer}.
First, we compute $E_{0,0}(a_1,a_2)$ as defined in the lemma, using the following code.

\begin{verbatim}
d = 1;
n = 3;

b[k_] := Which[k == 0, 1/2, k == 1, a1, 
		k == 2, a2, k == -1, 1 - a1, k == -2, 1 - a2];

pX[i_] := If[Abs[i] >= n, 1, b[i]];
pY[j_] := If[Abs[j] >= n, 0, b[j]];

states =   Select[Tuples[{Range[-n, n + 2 d], Range[-n - 2 d, n]}],
				 #[[1]] - #[[2]] <= 2 d &&  EvenQ[#[[1]] - #[[2]]] &];

abs = Select[states, #[[1]] - #[[2]] == 2 d &];
non = Complement[states, abs];

vars = Association[Table[s -> e[s[[1]], s[[2]]], {s, non}]];

eqs = Table[Module[{i = s[[1]], j = s[[2]], px, py, trans}, px = pX[i];
    py = pY[j];
    trans = {{{i + 1, j + 1}, px py}, {{i + 1, j - 1}, 
       px (1 - py)}, {{i - 1, j + 1}, (1 - px) py}, {{i - 1, 
        j - 1}, (1 - px) (1 - py)}};
    vars[s] == 
     1 + Total[
       If[KeyExistsQ[vars, #[[1]]], #[[2]] vars[#[[1]]], 0] & /@ trans]], 
       {s, non}];

sol = First@Solve[eqs, Values[vars]];

E00 = FullSimplify[Factor[e[0, 0] /. sol]]
\end{verbatim}

The output to the following code is \texttt{False}. 
It shows that for $E_{0,0}(a_1,a_2)$, there is no point
$(a_1,a_2)\in(0,1)^2$ at which both partial derivatives vanish and the
denominator is nonzero.
\begin{verbatim}
testmath[a1_, a2_] := 
(	-a1  (5 + a1  (11 + (-3 + a1)  a1)) - 14  (-1 + a2) 
	+  a1  (-7 + (-2 + a1)  a1  (-26 + a1^2))  a2 
	-   2  (-1 + a1)  (-7 +  a1  (27 + 2  a1  (-7 + 2  (-3 + a1)  a1)))  a2^2 
	+  2  (-1 + a1)  (-7 +  a1  (14 + a1  (15 + a1  (-29 + 10  a1))))  a2^3 
	- (-1 +  a1)^2  a1  (21 + a1  (-37 + 20  a1))  a2^4 
	+  7  (-1 + a1)^4  a1  a2^5
)
/(	(-1 + a1)  
	(	(-1 + a1)  (2 + a1  (5 + a1)) 
		+ (2 + a1  (7 + a1  (-7 + a1  (-13 + 3  a1))))  a2 
		-   2  (-1 + a1  (3 + a1  (6 + a1  (-14 + a1 + a1^2))))  a2^2 
		+  2  (-1 + a1)  (1 + a1 + (-2 + a1)  a1^2  (4 + 3  a1))  a2^3 
		- (-1 + a1)^2  a1  (-3 - 2  a1 + 6  a1^2)  a2^4 
		+ (-1 + a1)^3  a1  (1 + 2  a1)  a2^5
	)
);

Npoly = Numerator[Together[testmath[a1, a2]]];
Dpoly = Denominator[Together[testmath[a1, a2]]];

Epoly = FullSimplify[Dpoly/((a1 - 1) (a2 - 1))];

G1 = Expand[D[Dpoly, a1] Npoly - D[Npoly, a1] Dpoly];
G2 = Expand[D[Dpoly, a2] Npoly - D[Npoly, a2] Dpoly];

Reduce[{G1 == 0, G2 == 0, Epoly != 0, 0 < a1 < 1, 0 < a2 < 1}, 
  {a1, a2}, Reals, Method -> "CylindricalDecomposition"]
\end{verbatim}

\section{The Support of the Oblivious Strategies, for $d=1,\ldots,6$}
\label{sec:CandOptSupport}

\begin{table}[!h]
\centering
\scriptsize
\setlength{\tabcolsep}{3pt}
\begin{tabular}{c c c c c c c}
\hline
& $d=1$ & $d=2$ & $d=3$ & $d=4$ & $d=5$ & $d=6$ \\
\hline
$a_1$ & 0.199256999224 & 0.510592701206 & 0.512931318666 & 0.512722350241 & 0.511638743100 & 0.510482689385 \\
$a_2$ & 0 & 0.358187800391 & 0.479226061268 & 0.497610517542 & 0.503169833641 & 0.505130229674 \\
$a_3$ &  & 0.332727973539 & 0.423967871873 & 0.475390850025 & 0.491415084361 & 0.497945272494 \\
$a_4$ &  & 0.258214997695 & 0.391472964185 & 0.448089184416 & 0.476514466580 & 0.488952954463 \\
$a_5$ &  & 0 & 0.378469909722 & 0.425810342760 & 0.459931363441 & 0.478414859081 \\
$a_6$ &  &  & 0.345664995034 & 0.411628362935 & 0.445033080053 & 0.467193361811 \\
$a_7$ &  &  & 0.246783051357 & 0.405015811676 & 0.433647076748 & 0.456692278305 \\
$a_8$ &  &  & 0.051759013397 & 0.385764474436 & 0.426126189355 & 0.447918575447 \\
$a_9$ &  &  & 0 & 0.336371427261 & 0.422201209540 & 0.441288401885 \\
$a_{10}$ &  &  &  & 0.264259249658 & 0.409453194378 & 0.436801813261 \\
$a_{11}$ &  &  &  & 0.163241017065 & 0.379651156085 & 0.434158570964 \\
$a_{12}$ &  &  &  & 0.034589169986 & 0.339997988607 & 0.425016462575 \\
$a_{13}$ &  &  &  & 0 & 0.287848984107 & 0.404955658349 \\
$a_{14}$ &  &  &  &  & 0.221796997736 & 0.379757882382 \\
$a_{15}$ &  &  &  &  & 0.143701927823 & 0.348162173931 \\
$a_{16}$ &  &  &  &  & 0.052098019372 & 0.309113945143 \\
$a_{17}$ &  &  &  &  & 0 & 0.262259907796 \\
$a_{18}$ &  &  &  &  &  & 0.208138846566 \\
$a_{19}$ &  &  &  &  &  & 0.146486977763 \\
$a_{20}$ &  &  &  &  &  & 0.079495834116 \\
$a_{21}$ &  &  &  &  &  & 0 \\
\hline
\end{tabular}
\caption{Numerically computed strategy parameters for the weak-peek optimization.
For each fixed $d$, the column lists the optimized values of the parameters $a_i$ up to the first reported zero.  The vanishing index is therefore identified by the first zero appearing in each column.}
\label{tab:weakpeek-alphas}
\end{table}

\end{document}